\documentclass[10pt]{article}
\usepackage{amsmath,amssymb,amsthm,amsfonts,amsbsy,bm,bbm,latexsym,color}
\usepackage[letterpaper,margin=1.2in]{geometry}
\usepackage[english]{babel}
\usepackage{hyperref}
\usepackage{mathrsfs}
\usepackage{mathabx}
\usepackage{mathtools}
\usepackage{yfonts}

\usepackage{stmaryrd}

\newtheorem{theorem}{Theorem}[section]
\newtheorem{lemma}[theorem]{Lemma}

\newtheorem{remark}[theorem]{Remark}

\newtheorem{corollary}[theorem]  {Corollary}
\newtheorem{proposition}[theorem]{Proposition}

\newcommand{\id}{\mathbbm{1}}

\newcommand{\norm}[1]{\left\| #1 \right\|}
\newcommand{\scp}[2]{\big\langle #1 , #2 \big\rangle}

\newcommand{\bra}[1]{\langle #1 |}
\newcommand{\ket}[1]{| #1 \rangle}

\newcommand{\abs}[1]{| #1 |}
\newcommand{\tr}{\textnormal{Tr} \,}
\newcommand{\e}{\varepsilon}

\usepackage[normalem]{ulem}

\usepackage[colorinlistoftodos]{todonotes}
\begin{document}

\title{Propagation of moments for large data and semiclassical limit to the relativistic Vlasov equation}

\author{%Daniele Dimonte\footnote{daniele.dimonte@unibas.ch}, \ 
Nikolai Leopold\footnote{nikolai.leopold@unibas.ch} \ and Chiara Saffirio\footnote{chiara.saffirio@unibas.ch}
\\
University of Basel\footnote{Department of Mathematics and Computer Science, Spiegelgasse 1, 4051 Basel, Switzerland}
}

\maketitle

\begin{abstract}

We investigate the semiclassical limit from the semi-relativistic Hartree-Fock equation describing the time evolution of a system of fermions in the mean-field regime with a relativistic dispersion law and interacting through a singular potential of the form $K(x)=\gamma\frac{1}{|x|^a}$, $a \in \left( \max \left\{ \frac{d}{2} -2 , - 1 \right\},  d-2 \right]$, $d\in\{2,3\}$ and $\gamma\in\mathbb{R}$, with the convention $K(x)=\gamma\log(|x|)$ if $a=0$. For mixed states, we show convergence in Schatten norms with explicit rate towards the Weyl transform of a solution to the relativistic Vlasov equation with singular potentials, thus generalizing \cite{DRS2017} where the case of smooth potentials has been treated. Moreover, we provide new results on the well-posedness theory of the relativistic Vlasov equations with singular interactions.  
\end{abstract}

\noindent
\textbf{MSC class:} 35Q41, 82C10, 35Q55, 35Q83, 70H40, 81Q20

\frenchspacing

%\listoftodos

\section{Introduction}

We consider the time evolution of a large system of fermions with relativistic dispersion law and its mean-field approximation given by the {semi-}relativistic Hartree-Fock equation 
\begin{align}
\label{eq:Hartree-Fock equation relativistic}
i \varepsilon \partial_t \omega_{N,t} &= \left[ \sqrt{1 - \varepsilon^2 \Delta} + K * \rho_t -X_t, \omega_{N,t} \right]
\end{align}
where $\omega_{N,t}$ is a sequence of time-dependent self-adjoint operators acting on $L^2(\mathbb{R}^d)$ with $\tr  \omega_{N,t}= N$ and $0 \leq \omega_{N,t} \leq 1$. The semiclassical parameter $\varepsilon$ plays the role of the Planck constant $\hbar$ and depends on the number of particles $N$ as $\varepsilon = N^{- \frac{1}{d}}$, $\sqrt{1-\e^2\Delta}$ is the pseudodifferential operator defined by the multiplication by the symbol $\sqrt{1+\e^2|\xi|^2}$, $K:\mathbb{R}^d\to \mathbb{R}$ is the two-body interaction potential, $\rho_t(x) = N^{-1} \omega_{N,t}(x;x)$, where $\omega_{N,t}(\,\cdot\,;\,\cdot\,)$ denotes the kernel of the operator $\omega_{N,t}$, and $X_t$ is referred to as the exchange term, whose integral kernel is given by $X_t(x;y)=N^{-1}K(x-y)\,\omega_{N,t}(x;y)$. \\
Furthermore, we introduce the semi-relativistic Hartree equation
\begin{align}
\label{eq:Hartree equation relativistic}
i \varepsilon \partial_t \omega_{N,t} &= \left[ \sqrt{1 - \varepsilon^2 \Delta} + K * \rho_t, \omega_{N,t} \right]
\end{align}
obtained from equation~\eqref{eq:Hartree-Fock equation relativistic} by setting the exchange term $X_t$ to zero.

The aim of the paper is to present a rigorous analysis of the semiclassical limit $\e\to 0$ in the case of singular interactions, including the most physically interesting cases of the Coulomb and gravitational potentials in dimension two and three. More precisely,  let $d \in \{2, 3 \}$, $\gamma \in \mathbb{R}$, $a \in \left( - 1 ,  d-2 \right]$. The potential $K: \mathbb{R}^d \rightarrow \mathbb{R}^d$ is given by
\begin{align}
\label{eq:definition potential}
K(x) = \gamma
\begin{cases}
\ln \left( \abs{x} \right) \quad &\text{if} \, a =0
\\
\frac{1}{\abs{x}^a}  &\text{else}. 
\end{cases}
\end{align}

In this setting we prove (see Theorem~\ref{theorem:derivation relativistic Vlasov from Hartree}) that as $\e\to 0$ the dynamics described by \eqref{eq:Hartree-Fock equation relativistic} is well approximated by the relativistic Vlasov equation, the following nonlinear transport equation for the probability density $f: \mathbb{R}_{+} \times \mathbb{R}^d \times \mathbb{R}^d \rightarrow \mathbb{R}$
\begin{align}
\label{eq:relativistic Vlasov}
\partial_t f + \frac{v}{\sqrt{1 + v^2}} \cdot \nabla_x f + E \cdot \nabla_{v} f = 0
\end{align}
where $E = - \nabla K * \rho_f$ and $\rho_f(t,x) = \int_{\mathbb{R}^d}  f(t,x, v)\,dv $. 
If {$\gamma = \pm1$} and $(d, a) = \left\{ (2,0) , (3,1) \right\}$ the function $K$ is (modulo a constant) the Green function for the Poisson equation in $\mathbb{R}^d$. In these cases \eqref{eq:relativistic Vlasov} is called the relativistic Vlasov-Poisson system.
\smallskip

{\bf State of the art and strategy.}  The semiclassical limit from the semi-relativistic Hartree equation~\eqref{eq:Hartree equation relativistic} towards the relativistic Vlasov equation has been tackled in \cite{DRS2017} in the case of smooth interactions $K\in W^{2,\infty}(\mathbb{R}^3)$ and space dimension $d=3$ and in \cite{AMS2008} for $K(x)=\pm\frac{1}{|x|}$ and $d=3$. In \cite{DRS2017} the authors provide strong convergence with an explicit bound on the convergence rate. In \cite{AMS2008} weak convergence as $\e\to 0$ has been considered for Coulomb and gravitational interactions in the case of mixed states, but the methods of \cite{AMS2008} do not allow  for an explicit rate of convergence. The question of getting explicit control on the semiclassical approximation is not only of speculative nature, as real physical systems are made of a finite, although big, number of particles. It is therefore important for applications to specify how the semiclassical approximation depends on the number of particles. In this paper we generalise  \cite{AMS2008} in several directions: we show strong convergence in dimension $d=2,\,3$, exhibit an explicit rate, consider a larger class of interaction potentials and include the exchange term estimating it explicitly, thus proving the semiclassical limit from the semi-relativistic Hartree-Fock equation towards the relativistic Vlasov equation with singular interactions in two and three spatial dimensions in strong topology and with explicit rate.\\ We strongly rely on the method used in \cite[Theorem 2.1]{DRS2017}, generalising it to spatial dimension $d\in\{2,3\}$ and, most importantly, to singular interactions. The key idea is to adapt to the relativistic context estimates from \cite{LS2020}, where the semiclassical limit of the non-relativistic Hartree equation towards the non-relativistic Vlasov equation with singular potentials has been considered. The main difficulties to deal with come from the well-posedness theory and propagation of regularity for the relativistic Vlasov  equation, which is less understood than its non-relativistic counterparts. 
The semiclassical limit from the non-relativistic Hartree equation towards the Vlasov equation with regular and singular interaction potentials has indeed been extensively studied in \cite{LP1993,MM1993,GIMS1998,BGM2000,GMP2003,PP2009,APPP2011,FLP2012,AKN2013a,AKN2013b,BPSS2016,GMP2016,GP2017,GPP2018,GP2019,Saffirio2019,L2019,S2020,LS2020,L2021} and in \cite{NS1981,S1981,CLL2021}, where convergence was established directly starting from the many-body quantum evolution.

The well-posedness theory for the Vlasov equation is well understood even in the case of singular interactions and general initial data from the works of Lions and Perthame~\cite{LP1991} and Pfaffelmoser~\cite{P1992}. In the relativistic setting the case of smooth interactions is fully understood (see e.g.~\cite{DRS2017}), whereas for singular potentials, and in particular for Coulomb and gravitational interactions, the literature is somewhat limited. In the three dimensional attractive case ($\gamma=-1$, $a=1$ and $d=3$) global existence has been established in \cite{GS1985,GS2001,HR2007,KTZ2008,W2020} for radially symmetric initial data. We revisit these works and provide a well-posedness theory through propagation of velocity moments (global or local depending on the space dimension $d\in\{2,3\}$ and the parameter $a$) and propagation of regularity for general initial data in Theorem~\ref{theorem:velocity moments short time estimate} and Propositions~\ref{prop:regularity solution Vlasov} and~\ref{proposition: uniqueness Vlasov} . 

As for the {semi-}relativistic Hartree and Hartree-Fock equations, the case of pure states (namely density matrices which are projection operators) has been studied in \cite{FL2007,Le2007,HS, HLLS2010} and the case of mixed states has been outlined in \cite{AMS2008,HLLS2010}. In Proposition~\ref{lemma: Hartree equation global existence of solution} below we generalize these results to singular potentials of the form \eqref{eq:definition potential} in spatial dimension $d=3$.
%\tdn{Here the latest article of Enno must be included.}
%\tdn{Maybe we should mention that [25] and [1] treat the case of mixed states?}
 \smallskip

{\bf Notation and function spaces.} In order to establish a connection between solutions of the semi-relativistic Hartree-Fock equation, i.e. reduced densities, and solutions of the relativistic Vlasov equations, i.e. functions defined on phase-space, we define the Wigner transform of a one-particle reduced density matrix $\omega_N$  as
\begin{align}
W_N(x,v) = \left( \frac{\varepsilon}{2 \pi} \right)^d
\int_{\mathbb{R}^d}  \omega_{N} \left( x + \frac{\varepsilon y}{2} ; x - \frac{\varepsilon y}{2} \right) e^{- i v \cdot y}\,dy.
\end{align}
Its inverse is referred to as the Weyl quantization and given by
\begin{align}
\omega_N(x;y) &= \varepsilon^{-d} \int_{\mathbb{R}^d} W_N \left( \frac{x+y}{2} , v \right) e^{i v \cdot \frac{x-y}{\varepsilon}}\,dv .
\end{align}
For $\sigma \in \mathbb{N}_0$, $k \geq 0$ and $1 \leq p \leq \infty$ let $W_{k}^{\sigma,p} \left( \mathbb{R}^{d} \right)$ be the Sobolev space equipped with the norm 
\begin{align}
\label{eq:definition LP-spaces}
\norm{f}_{W_{k}^{\sigma, p} \left( \mathbb{R}^{d} \right)} 
&=
\begin{cases}
\left( \sum_{\abs{\alpha} \leq \sigma} \norm{\left< \cdot \right>^{k} D^{\alpha} f}_{L^p\left( \mathbb{R}^{d} \right)}^p \right)^{\frac{1}{p}}
\quad 1 \leq &p < \infty ,
\\
\max_{\abs{\alpha} \leq \sigma}  \norm{\left< \cdot \right>^{k} D^{\alpha} f}_{L^{\infty}\left( \mathbb{R}^{d} \right)}
\quad &p = \infty,
\end{cases}
\end{align}
where $\left< x \right>^2 = 1 + \abs{x}^2$.
In the cases $\sigma=0$ or $p=2$ we use the shorthand notation
$L_k^{p} \left( \mathbb{R}^{d} \right) \coloneqq W_{k}^{(0,p)} \left( \mathbb{R}^{d} \right)$
and  $H_k^{\sigma} \left( \mathbb{R}^{d} \right) \coloneqq W_{k}^{\sigma, 2} \left( \mathbb{R}^{d} \right)$. Vectors in $\mathbb{R}^{2d}$ are written as
 $z = (x,v)$ so that $\left< z \right>^2 = 1 + \abs{x}^2 + \abs{v}^2$ and $D_z^{\alpha} = \left( \partial/ \partial x_1 \right)^{\beta_1} \left( \partial/ \partial v_1 \right)^{\gamma_1} \cdots \left( \partial/ \partial x_d \right)^{\beta_d} \left( \partial/ \partial v_d \right)^{\gamma_d}$ with
$\beta = (\beta_i)_{i \in \llbracket 1, d \rrbracket} \in \mathbb{N}_0^d$ and $\gamma = (\gamma_i)_{i \in \llbracket 1, d \rrbracket} \in \mathbb{N}_0^d$ 
such that $\abs{\alpha} = \sum_{i=1}^d \beta_i + \gamma_i$. 
For two Banach spaces $A$ and $B$ we denote by $A \cap B$ the Banach space of vectors $f \in A \cap B$ with norm $\norm{f}_{A \cap B} = \norm{f}_A + \norm{f}_B$. We use $C$ to denote a generic positive constant that might depend on the dimension $d$, the strength of the singularity $a$ and the parameters $k$ and $\sigma$ appearing in \eqref{eq:definition LP-spaces}.

Let $\textfrak{S}^{\infty} \left( L^2(\mathbb{R}^d) \right)$ be the set of all bounded operators on $L^2 (\mathbb{R}^d)$ and $\textfrak{S}^1 \left( L^2(\mathbb{R}^d) \right)$ be the set of trace class operators on $L^2(\mathbb{R}^d)$. More generally, for $p\in[1,\infty)$, we denote by $\textfrak{S}^p(\mathbb{R}^d)$ the $p$-Schatten space equipped with the norm $\norm{A}_{\textfrak{S}^p}=(\tr \abs{A}^p)^\frac{1}{p}$, where $A$ is an operator, $A^*$ its adjoint and $\abs{A}=\sqrt{A^*A}$.

We define 
\begin{align}
\textfrak{S}^{1,\frac{1}{2}} \left( L^2(\mathbb{R}^d) \right) = \left\{
\omega: \omega \in \textfrak{S}^{\infty} \left( L^2(\mathbb{R}^d) \right), \; \omega^* = \omega \; \text{and} \;
\left( 1 - \Delta \right)^{1/4} \omega \left( 1 - \Delta \right)^{1/4} \in \textfrak{S}^1 \left( L^2(\mathbb{R}^d) \right)
\right\}
\end{align}
with norm
\begin{align}
\norm{\omega}_{\textfrak{S}^{1,\frac{1}{2}}}
&= 
\norm{ \left( 1 - \Delta \right)^{1/4} \omega  \left( 1 - \Delta \right)^{1/4}}_{\textfrak{S}^1} .
\end{align}
The positive cone of this space is defined as
\begin{align}
\textfrak{S}_{+}^{1,\frac{1}{2}} \left( L^2(\mathbb{R}^d) \right) = \left\{
\omega: \omega \in \textfrak{S}^{\infty} \left( L^2(\mathbb{R}^d) \right) 
\, , \omega \geq 0 \, ,
\; \text{and} \;
\left( 1 - \Delta \right)^{1/4} \omega \left( 1 - \Delta \right)^{1/4} \in \textfrak{S}^1 \left( L^2(\mathbb{R}^d) \right)
\right\} .
\end{align}
\smallskip

{\bf Main results.}  
   Our main result concerns the semiclassical limit from the {semi-}relativistic Hartree equation~\eqref{eq:Hartree equation relativistic} towards the relativistic Vlasov equation~\eqref{eq:relativistic Vlasov} with potentials of the form given in \eqref{eq:definition potential}.

\begin{theorem}[Semiclassical limit towards the relativistic Vlasov equation]
\label{theorem:derivation relativistic Vlasov from Hartree}
Let $d \in \{2 , 3 \}$, $a \in \left( \max \left\{\frac{d}{2} - 2, -1 \right\} , d- 2 \right]$, $\gamma\in\mathbb{R}$  and $K$ be defined as in \eqref{eq:definition potential}. Let  $\omega_N$ be a sequence of reduced density matrices on $L^2(\mathbb{R}^d)$, $\tr \omega_N= N$, $0 \leq \omega_N \leq 1$. 
Let $\omega_{N,t}$ be the unique solution of \eqref{eq:Hartree equation relativistic} with initial condition $\omega_{N,0} = \omega_N$. 
In addition, assume that
 $x \, \omega_{N} \, x \in   \textfrak{S}^1 \left( L^2(\mathbb{R}^d) \right)$ if $a \in \left( \max \left\{\frac{d}{2} - 2, -1 \right\} , 0 \right]$ and
 $\omega_{N} \in  \textfrak{S}_{+}^{1,\frac{1}{2}}\left( L^2(\mathbb{R}^d) \right) $ if $a \in (0,1]$.
Moreover, let $\widetilde{W}_{N,t}$ be the unique solution of \eqref{eq:relativistic Vlasov} with initial datum $\widetilde{W}_{N,0} = \widetilde{W}_N \geq 0$ verifying
\begin{align}
\label{eq: derivation relativistic Vlasov from Hartree condition Vlasov 1}
\widetilde{W}_{N,t} &\in L_{\rm loc}^{\infty} \left( \mathbb{R}_{+} , W^{3, \infty} (\mathbb{R}^{2d}) \cap H_{\sigma}^{\sigma + 1} (\mathbb{R}^{2d}) \right)
\\
\label{eq: derivation relativistic Vlasov from Hartree condition Vlasov 2}
\widetilde{\rho} &\in L_{\rm loc}^{\infty} \left( \mathbb{R}_{+} , L^1(\mathbb{R}^d) \cap H^{\nu} (\mathbb{R}^d) \right) ,
\end{align}
where $\nu = 4 + a  - d $ and  $\sigma= 4 + n$ with $n \in 2 \mathbb{N}$  such that $n \geq \frac{d (a+1)}{d - (a+1)}$.
For the Weyl quantization $\widetilde{\omega}_{N,t}$ of $\widetilde{W}_{N,t}$ it holds
\begin{align}
\label{eq:derivation relativistic Vlasov from Hartree}
\norm{\omega_{N,t} - \widetilde{\omega}_{N,t}}_{\textfrak{S}^1}
&\leq \left[ \norm{\omega_{N} - \widetilde{\omega}_{N}}_{\textfrak{S}^1}   + N \varepsilon  \right]
\left[ 
 C(t) + \int_0^t ds \, C(s) \lambda(s) e^{\int_s^t \lambda(\tau)\,d \tau}
\right] .
\end{align}
Here,
\begin{align}\label{eq:constant1}
\lambda(t) = C \abs{\gamma} \norm{\nabla_v \widetilde{W}_{N,s}}_{W^{2, \infty}(\mathbb{R}^{2d}) \cap H_{\sigma}^{\sigma}(\mathbb{R}^{2d})}
\end{align}
and
\begin{align}\label{eq:constant2}
C(t) &= 1 +  C \int_0^t \left( 1 + \abs{\gamma} \right) \Big( 1 
+ \norm{\widetilde{\rho}(s)}_{L^1(\mathbb{R}^d) \cap H^{\nu}(\mathbb{R}^d)} \Big) \norm{\widetilde{W}_{N,s}}_{H_{\sigma}^{\sigma} (\mathbb{R}^{2d})}\,ds,
\end{align}
where $C$ is a numerical constant depending only on the dimension $d$ and the parameters $n$, $a$. 
\end{theorem}

\noindent Some remarks are in order:

\begin{remark}
In $d=3$, if $\gamma < 0$ and $a = 1$ we assume in addition that the condition 
\eqref{eq: Hartree global well-posedness condition for attractive Coulomb} stated in Proposition~\ref{lemma: Hartree equation global existence of solution} for the well-posedness of the semi-relativistic Hartree-Fock equation holds true.
\end{remark}

\begin{remark}
Theorem~\ref{theorem:derivation relativistic Vlasov from Hartree} generalizes \cite[Theorem 2.1]{DRS2017} which considered the case $d = 3$ and regular interaction potentials $K \in W^{2, \infty} (\mathbb{R}^3)$. Notice that we obtain for singular interactions the same expected optimal rate of convergence as in the case of smooth potentials.
\end{remark}

\begin{remark}
Observe that thanks to \cite[Section 5]{LS2020}, we can include the exchange term to exhibit an explicit rate of convergence from the semi-relativistic Hartree-Fock equation~\eqref{eq:Hartree-Fock equation relativistic} towards the relativistic Vlasov equation as $\e\to 0$. More precisely we can bound the exchange term as in \cite[Proposition~5.2]{LS2020} (see estimate (42a)) and conclude that the rate of convergence towards the relativistic Vlasov equation does not change whether we consider the semi-relativistic Hartree or the Hartree-Fock equation as starting point.
\end{remark}

\begin{remark}
The dependence of $\lambda(t)$ and $C(t)$ on high regularity norms of the solution to the Vlasov equation restricts our analysis to the case of mixed states, namely one-particle density matrices of the form 
\begin{equation*}
\omega_{N,t}:=\sum_{j\in\mathbb{N}}\lambda_j|\psi_j(t)\rangle\langle\psi_j(t)|
\end{equation*}
where $\{\psi(t)\}_{j\in\mathbb{N}}$ is an orthonormal system of $L^2(\mathbb{R}^d)$, $\lambda_j\in \ell^1$ and $\lambda_j\geq 0$.
\end{remark}

\begin{remark}
Below it is shown (Theorem \ref{theorem:velocity moments short time estimate}, Corollary \ref{cor:estimates density} and Proposition \ref{prop:regularity solution Vlasov}) that \eqref{eq: derivation relativistic Vlasov from Hartree condition Vlasov 1}--\eqref{eq: derivation relativistic Vlasov from Hartree condition Vlasov 2} hold globally in time for $a \in (-1,1/2]$ and locally in time for $a \in (1/2, 1]$ if $W_N$ is regular enough. By means of \cite[Theorem 1.1.]{W2020}, Corollary \ref{cor:estimates density} and Proposition \ref{prop:regularity solution Vlasov} it is, moreover, possible to infer \eqref{eq: derivation relativistic Vlasov from Hartree condition Vlasov 1}--\eqref{eq: derivation relativistic Vlasov from Hartree condition Vlasov 2} globally in time for $a=1$ if $W_N$ is spherically symmetric and polynomially decaying as $(x,v) \rightarrow \infty$.
\end{remark}

Theorem~\ref{theorem:derivation relativistic Vlasov from Hartree} can be extended to general $p$-Schatten norm, with $p\in[1,\infty)$, in the same spirit of \cite[Theorem~1.2]{LS2020}. However, in the relativistic setting we have to deal with the kinetic term, that cannot be absorbed by unitary transformations (as it was the case in \cite[Theorem~1.2]{LS2020}), to obtain the following result.

%The analysis of this extra term is addressed in subsection~\ref{subsection:Derivation of the relativistic Vlasov equation}, where the proof of the following propRemark~\ref{remark:p-schatten-kinetic}. 

\begin{proposition}\label{remark:p-schatten}
%Let $d \in \{2 , 3 \}$, $a \in \left( \max \left\{\frac{d}{2} - 2, -1 \right\} , d- 2 \right]$, $\gamma\in\mathbb{R}$  and $K$ be defined as in \eqref{eq:definition potential}. Let  $\omega_N$ be a sequence of reduced density matrices on $L^2(\mathbb{R}^d)$, $\tr \omega_N= N$, $0 \leq \omega_N \leq 1$. 
%Let $\omega_{N,t}$ be the unique solution of \eqref{eq:Hartree equation relativistic} with initial condition $\omega_{N,0} = \omega_N$. 
%In addition, assume that
 %$x \, \omega_{N} \, x \in   \textfrak{S}^1 \left( L^2(\mathbb{R}^d) \right)$ if $a \in \left( \max \left\{\frac{d}{2} - 2, -1 \right\} , 0 \right]$ and
% $\omega_{N} \in  \textfrak{S}_{+}^{1,\frac{1}{2}}\left( L^2(\mathbb{R}^d) \right) $ if $a \in (0,1]$.
%Moreover, let $\widetilde{W}_{N,t}$ be the unique solution of \eqref{eq:relativistic Vlasov} with initial datum $\widetilde{W}_{N,0} = W_N \geq 0$ verifying
%\begin{align}
%\label{eq: derivation relativistic Vlasov from Hartree condition Vlasov 1}
%\widetilde{W}_{N,t} &\in L_{\rm loc}^{\infty} \left( \mathbb{R}_{+} , W^{n_0 + 1, \infty} (\mathbb{R}^{2d}) \cap H_{\sigma}^{\sigma + 1} (\mathbb{R}^{2d}) \right)
%\\
%\label{eq: derivation relativistic Vlasov from Hartree condition Vlasov 2}
%\widetilde{\rho} &\in L_{\rm loc}^{\infty} \left( \mathbb{R}_{+} , L^1(\mathbb{R}^d) \cap H^{\nu} (\mathbb{R}^d) \right) ,
%\end{align}
%where $n_0 = \lfloor d/2 \rfloor + 1$, $(n, n_1) \in \left( 2 \mathbb{N} \right)^2$ are such that $n > d/2$ and $n_1 \geq \frac{d}{b-1}$ and we used the notations $\sigma= 2 n + n_1$ and $\nu = \left( n + a + 2 - d \right)$ and $b = \frac{d}{a+1}$.
Under the same assumptions and notations of Theorem~\ref{theorem:derivation relativistic Vlasov from Hartree}, it holds
\begin{align*}
%\label{eq:derivation relativistic Vlasov from Hartree}
\norm{\omega_{N,t} - \widetilde{\omega}_{N,t}}_{\textfrak{S}^p}
&\leq \norm{\omega_{N}-\widetilde{\omega}_{N}}_{\textfrak{S}^p}+\left[ \norm{\omega_{N} - \widetilde{\omega}_{N}}_{\textfrak{S}^1}   + C(t)N \varepsilon  \right]
 e^{\int_s^t \lambda(\tau)\,d \tau}
\end{align*}
for $p\in \left[1, \min \left\{ d/(a+1) , 2 \right\} \right)$, and
\begin{align*}
%\label{eq:derivation relativistic Vlasov from Hartree}
\norm{\omega_{N,t} - \widetilde{\omega}_{N,t}}_{\textfrak{S}^p}
&\leq C(t)\left(\norm{\omega_{N}-\widetilde{\omega}_{N}}_{\textfrak{S}^q}^\frac{q}{p}+ \norm{\omega_{N} - \widetilde{\omega}_{N}}_{\textfrak{S}^1}^\frac{q}{p}   + N^\frac{1}{p}\e^\frac{q}{p}\right) 
 e^{\frac{q}{p}\int_s^t \lambda(\tau)\,d \tau}
\end{align*}
for $p\in[\min \left\{ d/(a+1) , 2 \right\} ,\infty)$ and $q\in[1,\min \left\{ d/(a+1) , 2 \right\})$.
%Here,
%\begin{align}\label{eq:constant1}
%\lambda(t) = C_{d,n_1,a} \norm{\nabla_v \widetilde{W}_{N,s}}_{W^{n_0, \infty}(\mathbb{R}^{2d}) \cap H_{\sigma}^{\sigma}(\mathbb{R}^{2d})}
%\end{align}
%and
%\begin{align}\label{eq:constant2}
%C(t) &=  C_{d,n_1,a} \int_0^t \Big( 1 
%+ \norm{\widetilde{\rho}(s)}_{L^1(\mathbb{R}^d) \cap H^{\nu}(\mathbb{R}^d)} \Big) \norm{\widetilde{W}_{N,s}}_{H_{\sigma}^{\sigma} (\mathbb{R}^{2d})}\,ds,
%\end{align}
%where $C_{d,n_1,a}$ is a numerical constant depending only on the dimension $d$ and the parameters $n_1$, $a$. 
\end{proposition}

In particular, in Proposition~\ref{remark:p-schatten} we obtain an explicit rate for $\omega_{N,t}-\widetilde{\omega}_{N,t}$ in Hilbert-Schmidt norm (i.e. $\textfrak{S}^2$), thus generalising \cite[Theorem~2.3]{DRS2017} to singular interactions of the form \eqref{eq:definition potential}.

In order to provide a rigorous approximation of the {semi-}relativistic Hartree equation in terms of the relativistic Vlasov equation, we need both equations~\eqref{eq:Hartree equation relativistic} and~\eqref{eq:relativistic Vlasov} to be well-posed and to satisfy some regularity properties. For these reasons our next results  deal with propagation of velocity moments and regularity for the relativistic Vlasov equation~\eqref{eq:relativistic Vlasov} (see Theorem~\ref{theorem:velocity moments short time estimate} and Proposition~\ref{prop:regularity solution Vlasov}), conditions for uniqueness for the Cauchy problem associated with \eqref{eq:relativistic Vlasov}  (see Proposition~\ref{proposition: uniqueness Vlasov}) and well-posedness of the semi-relativistic Hartree-Fock equation~\eqref{eq:Hartree-Fock equation relativistic} with singular interaction (see Proposition~\ref{lemma: Hartree equation global existence of solution}), thus providing conditions for \eqref{eq:constant1} and \eqref{eq:constant2} to be finite and characterising a class of states for which Theorem~\ref{theorem:derivation relativistic Vlasov from Hartree} holds true. The proofs are postponed to Section~\ref{section:proofs}.

\begin{theorem}[Velocity moments]
\label{theorem:velocity moments short time estimate}
Let $d \in \{ 2, 3 \}$, $a \in (-1, d-2]$, $\gamma\in\mathbb{R}$ and $K$ be defined as in \eqref{eq:definition potential}. Let $f_0 \geq 0$, $f_0 \in L^1 \cap L^{\infty} \left( \mathbb{R}^d \times \mathbb{R}^d \right)$.
\begin{enumerate}

\item Let $a \in (-1, 1/2]$ and assume that 
\begin{align}
\int_{\mathbb{R}^d \times \mathbb{R}^d}  \abs{v }^k f_0(x,v)\,dx \, dv   <  + \infty
\quad \text{for} \; k \in \mathbb{N}.
\end{align}
Then, for all $1 \leq  p < + \infty$ there exists a solution $f \in C \left( \mathbb{R}_{+} , L^p \left( \mathbb{R}^d \times \mathbb{R}^d \right) \right) \cap L^{\infty} \left( \mathbb{R}_{+} , L^{\infty} \left( \mathbb{R}^d \times \mathbb{R}^d \right) \right)$ of \eqref{eq:relativistic Vlasov} satisfying
\begin{align}
\label{eq:velocity moments short time estimate}
\sup_{t \in [0,T]} \int_{\mathbb{R}^d \times \mathbb{R}^d}  \abs{v }^k f(t,x,v)\,dx \, dv   < + \infty
\quad \text{for all } \;  T < + \infty .
\end{align}

\item Let $d=3$, $a \in (1/2, 1]$ and assume
\begin{align}
\int_{\mathbb{R}^3 \times \mathbb{R}^3} \abs{v }^k f_0(x,v)\,dx \, dv   <  + \infty
\quad \text{for} \; k \in \mathbb{N}  \; \text{such that} \; k \geq \frac{d a}{d - (a+1)}.
\end{align}
Then, for all $1 \leq  p < + \infty$ there exists  $T \in \mathbb{R}_{+}$ and a solution $f \in C \left([0,T] , L^p \left( \mathbb{R}^3 \times \mathbb{R}^3 \right) \right) \cap L^{\infty} \left( [0,T] , L^{\infty} \left( \mathbb{R}^3 \times \mathbb{R}^3 \right) \right)$ of \eqref{eq:relativistic Vlasov} satisfying 
\begin{align}
\label{eq:velocity moments short time estimate}
\sup_{t \in [0,T]} \int_{\mathbb{R}^3 \times \mathbb{R}^3}  \abs{v }^k f(t,x,v)\,dx \, dv   < + \infty .
\end{align}
\end{enumerate}
\end{theorem}

\begin{corollary}
\label{cor:estimates density}
Let $f$ be a solution of \eqref{eq:relativistic Vlasov} satisfying the assumptions of Theorem \ref{theorem:velocity moments short time estimate}. In addition, let $f_0 \abs{v}^n \in L^{\infty} \left( \mathbb{R}^d \times \mathbb{R}^d \right)$ for $n > d$ and {let the parameter $k$ of Theorem \ref{theorem:velocity moments short time estimate} satisfy $k > \frac{d (a+1)}{d - (a+1)}$}. Then there exists a time $T \in \mathbb{R}_{+}$ such that 
\begin{align}
E &\in L^{\infty} \left( [0,T] , L^{\infty}(\mathbb{R}^d) \right) ,
\\
\rho_f &\in L^{\infty} \left( [0,T],   L^1 \cap L^{\infty} \left( \mathbb{R}^d \times \mathbb{R}^d \right) \right) .
\end{align}
\end{corollary}

\begin{proposition}[Propagation of regularity]
\label{prop:regularity solution Vlasov}
Let $d \in \{ 2, 3 \}$, $a \in (-1, d-2]$, {$\gamma \in\mathbb{R}$}, $K$ be defined as in \eqref{eq:definition potential}. Let $(n, \sigma) \in \mathbb{N}^2$ be such that $n> 2 d$ and $f \geq 0$ be a solution of the relativistic Vlasov equation \eqref{eq:relativistic Vlasov} with initial data $f_0 \in W_n^{\sigma, \infty}(\mathbb{R}^{2d})$.
Moreover, assume that  $\rho_f \in L_{\rm loc}^{\infty} \left( \mathbb{R}_+ ,  L^1 \cap L^{\infty} \left( \mathbb{R}^d \times \mathbb{R}^d \right) \right)$.
Then the following regularity estimates hold
\begin{align}
\label{eq:regularity estimate 1}
f \in L_{\rm loc}^{\infty} \big( \mathbb{R}_+ , W_n^{\sigma, \infty}(\mathbb{R}^{2d}) \big) 
\\
\label{eq:regularity estimate 2}
\nabla^{\sigma} \rho_f \in L_{\rm loc}^{\infty} \big( \mathbb{R}_+ , L^{\infty}(\mathbb{R}^d) \big) .
\end{align}
If in addition $f_0 \in H_k^{\sigma}(\mathbb{R}^{2d})$ for some $k \in \mathbb{R}_+$, then
\begin{align}
\label{eq:regularity estimate 3}
f \in L_{\rm loc}^{\infty} \big( \mathbb{R}_+ , H_k^{\sigma} (\mathbb{R}^{2d}) \big) .
\end{align}
Note that \eqref{eq:regularity estimate 1}--\eqref{eq:regularity estimate 3} hold globally and not only locally if $\rho_f \in L^{\infty} \left( \mathbb{R}_{+},   L^1 \cap L^{\infty} \left( \mathbb{R}^d \times \mathbb{R}^d \right) \right)$.
\end{proposition}

\begin{proposition}[Uniqueness criterion]
\label{proposition: uniqueness Vlasov}
Let $d \in \{2, 3 \}$ and $a \in (-1, d-2]$. Moreover, let $f_1$ and $f_2$ be two solutions of the Vlasov equation \eqref{eq:relativistic Vlasov} in $L^{\infty} \left( [0,T] , L^1(\mathbb{R}^{2d}) \right)$ for some $T>0$ such that
\begin{align}
\label{eq:uniqueness Vlasov condition}
\int_{\mathbb{R}^d} \abs{\nabla_v f_2}\,dv  \in L^1 \left( [0,T] , L^{\frac{d}{d- (a+1)} - \delta} \cap L^{\frac{d}{d- (a+1)} + \delta}(\mathbb{R}^d) \right) 
\end{align} 
for some  $0 < \delta < \frac{d}{d - (a+1)}$.
Then there exists  $C>0$ such that
\begin{align}
\label{eq:uniqueness Vlasov}
\norm{(f_1 - f_2)(t)}_{L^1(\mathbb{R}^{2d})}
&\leq C \norm{(f_1 - f_2)(0)}_{L^1(\mathbb{R}^{2d})} e^{C \int_0^t \norm{\rho_{\abs{\nabla_v f_2}}(s)}_{L^{\frac{d}{d - (a+1)} + \delta} \cap L^{\frac{d}{d - (a+1)} - \delta} (\mathbb{R}^d)}\,ds},
\end{align}
where $\rho_{\abs{\nabla_v f_2}}(t,x):=\int_{\mathbb{R}^d} \abs{\nabla_v f_2(t,x,v)}\,dv$.
\end{proposition}

\begin{remark}\label{remark:uniqueness}
Note that \eqref{eq:uniqueness Vlasov condition} can be replaced by the stronger condition $f_2 \in L^1 \left( [0,T] , W_{n}^{1, \infty}(\mathbb{R}^{2d}) \right)$ with $n > 2d$. Indeed, note that by interpolation it holds
\begin{align}
\norm{\rho_{\abs{\nabla_v f_2}}(t)}_{L^{\frac{d}{d - (a+1)} + \delta} \cap L^{\frac{d}{d - (a+1)} - \delta} (\mathbb{R}^d)}
&\leq C \left(\norm{\rho_{\abs{\nabla_v f_2}}(t)}_{L^{1}  (\mathbb{R}^d)}
+ \norm{\rho_{\abs{\nabla_v f_2}}(t)}_{L^{\frac{d}{d - (a+1)} + \delta} (\mathbb{R}^d)}\right).
\end{align}
 Furthermore, we use the assumption $n/2 > d$ to estimate
\begin{align}
\norm{\rho_{\abs{\nabla_v f_2}}(t)}_{L^{1}  (\mathbb{R}^d)} 
&\leq \norm{\left< x \right>^{n/2} \left< v \right>^{n/2} \abs{\nabla_v f_2}(t)}_{L^{\infty}(\mathbb{R}^{2d})} \int_{\mathbb{R}^d \times \mathbb{R}^d} \left< x \right>^{- n/2} \left< v \right>^{- n/2}\,dx\,dv 
\nonumber \\
&\leq C \norm{f_2(t)}_{W_{n}^{1,\infty}(\mathbb{R}^{2d})}.
\end{align} 
Similarly, %\textcolor{blue}{let $k:={\frac{d}{d - (a+1)} + \delta}$}
\begin{align}
\norm{\rho_{\abs{\nabla_v f_2}}(t)}_{L^{\frac{d}{d - (a+1)} + \delta} (\mathbb{R}^d)}
%&= \int_{\mathbb{R}^d}  \Big| \int_{\mathbb{R}^d} \abs{\nabla_v f_2(t,x,v)}\,dv \Big|^{k}\,dx
%\nonumber \\
%&\leq 
% \norm{\left< x \right>^{\frac{n}{2}} \left< v \right>^{\frac{n}{2}} \abs{\nabla_v f_2}(t)}_{L^{\infty}(\mathbb{R}^{2d})}^{\textcolor{blue}{k}}
% \int_{\mathbb{R}^d}  \left< x \right>^{- \frac{n}{2} k }\,dx\,\left( \int_{\mathbb{R}^d}  \left< v \right>^{- \frac{n}{2} }\,dv\right)^{k}
% \nonumber \\
&\leq C \norm{f_2}_{W_{n}^{1,\infty}(\mathbb{R}^{2d})}\end{align}
%\tdn{Maybe we do not need to give the details of the second estimate?}
because $n/2 > d$ and ${\frac{d}{d - (a+1)} + \delta} \geq 1$ for $a \in (-1, d-2]$.
%\begin{align}
%\xout{\frac{n}{2} \left[  \frac{d}{d - (a+1)} + \delta \right] > d \left[  \frac{d}{d - (a+1)} + \delta \right]  > \frac{d^2}{d - (a+1)}  > d .}
%\end{align}
This shows that the solutions considered in Proposition \ref{prop:regularity solution Vlasov} are unique.
\end{remark}

%\tdc{In general I don't see a problem in considering Hartree-Fock keeping the exchange term in the analysis. Did you find a problem there?} 
%\begin{align}
%\label{eq:Hartree equation relativistic}
%i \varepsilon \partial_t \omega_{N,t} &= \left[ \sqrt{1 - \varepsilon^2 \Delta} + K * \rho_t , \omega_{N,t} \right]
%\end{align}
%with initial data $\omega_{N,0} = \omega_N$ and $\rho_t(x) = N^{-1} \omega_{N,t}(x;x)$, where $\omega_{N,t}(\,\cdot\,;\,\cdot\,)$ denotes the kernel of the operator $\omega_{N,t}$. The semiclassical parameter $\varepsilon$ is dependent on the number of particles and scales as $\varepsilon = N^{- \frac{1}{d}}$.
% In this work we are interested in fermionic mixed states normalized to $N = \varepsilon^{-d}$, i.e. $\omega_N$ is a sequence of reduced densities on $L^2(\mathbb{R}^d)$ with $\tr  \omega_N = N$ and $0 \leq \omega_N \leq 1$. 

%\tdn{I should adapt the assumptions of the Theorem: The best option to me seems to state a conditional result and later explain in the remarks under which circumstances the requirements are fulfilled? In this way we could still state the result in two dimensions and leave the well-posedness of the Hartree equation as an open problem?}

We will rely on the following well-posedness result generalizing to general singular potentials \cite[Theorem 1]{HLLS2010} and \cite[Lemma 2.3 and Lemma 2.4]{AMS2008} which considered the case $d=3$ and $a=1$. For the sake of completeness we provide its proof in Appendix \ref{section:Well-posedness of the Hartree-Fock equations}.
Note that we deliberately choose to restrict the well-posedness result  to $d=3$, 
$
K(x) = \gamma \frac{1}{\abs{x}^a} ,
$ 
 $0 < a \leq 1$ and coupling constant $\gamma \in \mathbb{R}$, in order to keep the presentation shorter. The cases $d=2$ or $a<0$ would require propagation of spatial moments, making the presentation much longer.

\begin{proposition}
\label{lemma: Hartree equation global existence of solution}
Let $d= 3$, $a \in (0,1]$ and $\omega_{N,0} \in \textfrak{S}^{1,\frac{1}{2}}_{+} \left( L^2(\mathbb{R}^3) \right)$ such that $0 \leq \omega_{N,0} \leq 1$.  If $\gamma < 0$ and $a = 1$ assume in addition that 
\begin{align}
\label{eq: Hartree global well-posedness condition for attractive Coulomb}
\abs{\gamma} <  \frac{N \varepsilon}{\gamma_{\rm cr}  \, \left( \tr \left( \omega_{N,0} \right)\right)^{2/3} }   ,
\end{align}
where $\gamma_{\rm cr}$ is a universal constant of order $1$. Then, the Cauchy problems for (the integral versions of) equation \eqref{eq:Hartree-Fock equation relativistic} and equation \eqref{eq:Hartree equation relativistic} have a unique global solution in $\textfrak{S}^{1,\frac{1}{2}}_{+} \left( L^2(\mathbb{R}^3) \right)$.
\end{proposition}

\begin{remark}
Without condition \eqref{eq: Hartree global well-posedness condition for attractive Coulomb} equation \eqref{eq:Hartree-Fock equation relativistic} and \eqref{eq:Hartree equation relativistic} are still locally well-posed in $\textfrak{S}^{1,\frac{1}{2}}_{+} \left( L^2(\mathbb{R}^3) \right)$ for $d=3$, $a=1$ and for all $\gamma \in \mathbb{R}$.
\end{remark}

The paper is structured as follows: in Section~\ref{section:preliminary-estimate} we collect estimates on the relativistic Vlasov equation and the semi-relativistic Hartree equation; Section~\ref{section:proofs} is devoted to the detailed proofs of semiclassical limit (Theorem~\ref{theorem:derivation relativistic Vlasov from Hartree}), propagation of velocity moments (Theorem~\ref{theorem:velocity moments short time estimate}), propagation of regularity (Proposition~\ref{prop:regularity solution Vlasov}) and uniqueness for the relativistic Vlasov equation with singular interactions (Proposition~\ref{proposition: uniqueness Vlasov});
Appendices \ref{section:Well-posedness of the Hartree-Fock equations} and \ref{appendix:duhamel} are concerned with the well-posedness theory for the Hartree and Hartree-Fock equations (Proposition~\ref{lemma: Hartree equation global existence of solution}) and address the existence of the two-parameter semigroup used in the Duhamel expansion in the proof of Theorem~\ref{theorem:derivation relativistic Vlasov from Hartree}. 

\section{Preliminary estimates}\label{section:preliminary-estimate}

Before we prove our main results, let us state some preliminary facts. 

\begin{proposition}
\label{prop:estimate velocity moments}
Let $d \in \mathbb{N}$, $c \geq b \geq 0$ and $f \coloneqq f(x,v) \geq 0$. Then for all $x \in \mathbb{R}^d$,
\begin{align}
\int_{\mathbb{R}^d}  \abs{v}^b f(x,v)\,dv \leq C \norm{f}_{L^\infty(\mathbb{R}^{2d})}^{\frac{c-b}{d+c}} \left( \int_{\mathbb{R}^d}  \abs{v}^c f(x,v)\,dv \right)^{\frac{d+b}{d+c}},
\end{align}
with $C$ depending only on $b$, $c$ and $d$. In particular, setting $\rho(x) = \int_{\mathbb{R}^d}  f(x,v)\,dv$, we have for any $b > 0$
\begin{align}
\norm{\rho}_{L^{\frac{b+d}{d}}(\mathbb{R}^d)}
&\leq C \norm{f}_{L^{\infty}(\mathbb{R}^{2d})}^{\frac{b}{d+b}}
\left( \int_{\mathbb{R}^d \times \mathbb{R}^d}  \abs{v}^{b} f(x,v)\,dx\,dv \right)^{\frac{d}{b+d}}.
\end{align}
\end{proposition}

\begin{proof}
The inequalities are proven in analogy to \cite[Proposition 2.1]{DMS2015} and the proof of estimate (14) in \cite{LP1991}.
\end{proof}

\begin{lemma}
\label{lemma:definition and interpolation of velocity moments}
Let $d \in \mathbb{N}$, $l \geq 0$, $f \coloneqq f(x,v) \geq 0$ and $M_l[f] = \int_{\mathbb{R}^d \times \mathbb{R}^d}  \abs{v}^l f(x,v)\,dx \, dv$. For $0 \leq b \leq a \leq c$ the following holds
\begin{align}
M_a[f] \leq M_b^{\frac{c-a}{c-b}}[f] \, M_c^{\frac{a-b}{c-b}}[f] .
\end{align}
\end{lemma}

\begin{proof}
Let $\mu = \frac{b(c-a)}{c-b}$, $\nu= \frac{c(a-b)}{c-b}$, $p = \frac{c-b}{c-a}$ and $q = \frac{c-b}{a-b}$. Then $\mu, \nu \geq 0$ and 
\begin{align}
\mu + \nu = \frac{b(c-a) + c (a-b)}{c-b} = \frac{bc - ab + ac - bc}{c-b} = \frac{a(c-b)}{c-b} = a .
\end{align}
Moreover, $p,q \geq 1$ and
\begin{align}
\frac{1}{p} + \frac{1}{q}
= \frac{c-a}{c-b} + \frac{a-b}{c-b} = \frac{c - a + a - b}{c-b}
= \frac{c-b}{c-b} = 1 .
\end{align}
Since $\mu p = \frac{b(c-a)}{(c-b)} \frac{(c-b)}{(c-a)} = b$ and
$\nu q = \frac{c(a-b)}{(c-b)} \frac{(c-b)}{(a-b)} = c$, H\"older's inequality leads to
\begin{align}
M_a[f] &= \int_{\mathbb{R}^d \times \mathbb{R}^d}   \abs{v}^{\mu} f^{\frac{1}{p}}(x,v) \abs{v}^{\nu} f^{\frac{1}{q}}(x,v) \,dx \, dv
\nonumber \\
&\leq  \left( \int_{\mathbb{R}^d \times \mathbb{R}^d}  \abs{v}^{\mu p} f(x,v) \,dx \, dv \right)^{\frac{1}{p}}
\left( \int_{\mathbb{R}^d \times \mathbb{R}^d}   \abs{v}^{\nu q} f(x,v) \,dx \, dv \right)^{\frac{1}{q}} 
\nonumber \\
&\leq  \left( \int_{\mathbb{R}^d \times \mathbb{R}^d}  \abs{v}^{b}  f(x,v) \,dx \, dv \right)^{\frac{c-a}{c-b}}
\left( \int_{\mathbb{R}^d \times \mathbb{R}^d}  \abs{v}^{c} f(x,v) \,dx \, dv \right)^{\frac{a-b}{c-b}} .
\end{align}
\end{proof}

\begin{lemma}
\label{lemma:auxiliary estimates}
Let $d \in \{2, 3 \}$, {$\gamma\in\mathbb{R}$}, $a \in (-1, d-2]$ and $K$ be defined as in \eqref{eq:definition potential}. Let $f \coloneqq f(x,v) \geq 0$, $\rho_f = \int_{\mathbb{R}^d} f(x,v)\,dv$, $E = - \nabla K * \rho_f$, $n \in \mathbb{N}$ be such that $n> 2d$, $p > \frac{d}{d - (a+1)}$ and  $1 \leq q < \frac{d}{d - (a+1)}$. Then there exists $C >0$ such that 
\begin{align}
\label{eq: auxiliary estimates electric field 1}
\norm{E}_{L^\infty(\mathbb{R}^d)} &\leq C \abs{\gamma} \left(\norm{\rho_f}_{L^p(\mathbb{R}^d)} + \norm{\rho_f}_{L^q(\mathbb{R}^d)} \right) \\
\label{eq: auxiliary estimates electric field 2}
\norm{\nabla E}_{L^\infty(\mathbb{R}^d)} &\leq C \abs{\gamma} \left( 1 +\norm{\rho_f}_{L^1(\mathbb{R}^d)} + \norm{\rho_f}_{L^\infty(\mathbb{R}^d)} \big[ 1 +  \ln (1 + \norm{\nabla \rho_f}_{L^\infty(\mathbb{R}^d)}) \big] \right) ,
\\
\label{eq: auxiliary estimates electric field 3}
\norm{E}_{W_0^{\sigma, \infty}(\mathbb{R}^d)} 
&\leq C \abs{\gamma} \min \left\{ \norm{f}_{W_n^{\sigma, \infty}(\mathbb{R}^{2d})} , \left( 1 +  \norm{f}_{W_n^{\sigma -1, \infty}(\mathbb{R}^{2d})} \right)  \left( 1 + \ln \big( 1 + \norm{f}_{W_n^{\sigma, \infty}(\mathbb{R}^{2d})} \big) \right) \right\} .
\end{align}
For $ \frac{d}{a+1} < q < \infty$ and $1 + \frac{1}{q} = \frac{a+1}{d} + \frac{1}{p}$ there exists $C>0$ such that 
\begin{align}
\label{eq: auxiliary estimates electric field 4}
\norm{E}_{L^q(\mathbb{R}^d)}  &\leq C \abs{\gamma} \norm{\rho_f}_{L^p(\mathbb{R}^d)}  
\leq C \abs{\gamma} \norm{f}_{L^{\infty}(\mathbb{R}^{2d})}^{\frac{p-1}{p}}
\left( \int_{\mathbb{R}^d \times \mathbb{R}^d}  \abs{v}^{(p-1)d} f(x,v)\,dx \, dv \right)^{\frac{1}{p}} .
\end{align}
\end{lemma}

\begin{proof}[Proof of Lemma \ref{lemma:auxiliary estimates}]
The first inequality is obtained by the splitting $K:=K\id_{|y|\leq 1}+K\id_{|y|>1}$ and Young's inequality. 
Similarly, one derives 
$
\norm{\nabla E}_{L^\infty(\mathbb{R}^d)}
\leq C \abs{\gamma} \big( \norm{\rho_f}_{L^\infty(\mathbb{R}^d)} + \norm{\rho_f}_{L^1(\mathbb{R}^d)} \big) 
$
for $a \in (-1, d-2)$ because
$\abs{\partial_{x_i} \partial_{x_j} K(x)} \leq  C \abs{\gamma} \frac{1}{\abs{x}^{a+2}} $. 
Inequality \eqref{eq: auxiliary estimates electric field 2} with $a = d-2$ is proven in \cite[p.83]{G2013}.  Note that
$D_x^{\alpha} E
= - \left( \nabla K \right) * D_x^{\alpha} \rho_f$. Together with \eqref{eq: auxiliary estimates electric field 2}
we get
\begin{align}
\norm{\nabla D_x^{\alpha} E}_{L^\infty(\mathbb{R}^d)} 
&\leq C \abs{\gamma} \left( 1 + \norm{D_x^{\alpha} \rho_f}_{L^1(\mathbb{R}^d)} + \norm{D_x^{\alpha} \rho_f}_{L^\infty(\mathbb{R}^d)} \left[ 1 + \ln \left( 1 + \norm{\nabla D_x^{\alpha} \rho_f}_{L^\infty(\mathbb{R}^d)} \right) \right] \right) .
\end{align}
Since $n > 2d$ we have  
\begin{align}
\norm{D_x^{\alpha} \rho_f}_{L^\infty(\mathbb{R}^d)} &\leq  \sup_{x} \int_{\mathbb{R}^d} \abs{D_x^{\alpha} f(t,x,v)}\,dv 
\leq  \norm{f}_{W_{n/2}^{\abs{\alpha}, \infty}(\mathbb{R}^{2d})} \int_{\mathbb{R}^d} \left< v \right>^{- \frac{n}{2}}\,dv
\leq C   \norm{f}_{W_{n/2}^{\abs{\alpha}, \infty}(\mathbb{R}^{2d})}  
\end{align}
and 
\begin{align}
\norm{D_x^{\alpha} \rho_f}_{L^1(\mathbb{R}^d)}
&= \int_{\mathbb{R}^d}  \left|{ \int_{\mathbb{R}^d}  D_x^{\alpha} f(t,x,v)\,dv}\right|\,dx
\leq 
\norm{f}_{W_n^{\abs{\alpha}, \infty}(\mathbb{R}^{2d})} \int_{\mathbb{R}^{2d}}  \left< z \right>^{-n}\,dz
\leq C  \norm{f}_{W_n^{\abs{\alpha}, \infty}(\mathbb{R}^{2d})}  .
\end{align}
In total, this shows
\begin{align}
\norm{E}_{W_0^{\sigma, \infty}(\mathbb{R}^d)}
&\leq C \abs{\gamma} \Big( 1 + \norm{f}_{W_n^{\sigma -1, \infty}(\mathbb{R}^{2d})} \Big) \Big( 1 +  \ln \big( 1 + \norm{f}_{W_n^{\sigma, \infty}(\mathbb{R}^{2d})} \big) \Big)  .
\end{align}
By similar means and \eqref{eq: auxiliary estimates electric field 1} one obtains
$\norm{E}_{W_0^{\sigma, \infty}(\mathbb{R}^d)}
\leq C \abs{\gamma}  \norm{f}_{W_n^{\sigma, \infty}(\mathbb{R}^{2d})} $.
In order to show the first inequality of \eqref{eq: auxiliary estimates electric field 4} we use that  $\abs{\nabla K}(x) \leq  \abs{\gamma}\abs{a} \abs{x}^{-(a+1)} \in L^{\frac{d}{a+1}, \infty}(\mathbb{R}^d)$, where $L^{p,\infty}(\mathbb{R}^d)$ denotes the weak-$L^p$ space of all measurable functions $f$ such that
$
\sup_{\alpha > 0} \left\{ \alpha \left( \int_{\mathbb{R}^d}  \id_{\abs{f(x)} > \alpha}\,dx \right)^{\frac{1}{p}} \right\} < \infty$
 (see, e.g., \cite[p. 106]{LL2001}). By means of the weak Young inequality we obtain
\begin{align}
\norm{E}_{L^q(\mathbb{R}^d)}
&= \norm{\nabla K * \rho_f}_{L^q(\mathbb{R}^d)}
\leq C \abs{\gamma} \norm{\rho_f}_{L^p(\mathbb{R}^d)}
\end{align}
where $1 + \frac{1}{q} = \frac{a+1}{d} + \frac{1}{p}$. Note that $p >1$ because $q > \frac{d}{a+1}$. Writing $p$ as $p = \frac{b+d}{d}$ with $b >0$ and applying Proposition \ref{prop:estimate velocity moments} leads to the second inequality in \eqref{eq: auxiliary estimates electric field 4}.
\end{proof}

\begin{proposition}
\label{proposition: Conservation laws relativistic Vlasov}
For a solution $f$ of \eqref{eq:relativistic Vlasov} with initial datum $f(0) = f_0 \geq 0$  the energy %\tdc{Why don't we state it for $d\in\{2,3\}$?}
\begin{align}
\mathcal{E}[f(t)] =  \int_{\mathbb{R}^d \times \mathbb{R}^d}  \left< v \right> f(t,x,v)\,dx \, dv + \frac{1}{2} \int_{\mathbb{R}^d \times \mathbb{R}^d} K(x-y) \rho_f(t,x) \rho_f(t,y)\,dx \, dy 
\end{align} is conserved during the time evolution, i.e. $\mathcal{E}[f(t)] = \mathcal{E}[f_0]$ for all $t \in \mathbb{R}$.
Moreover,
$\norm{\rho_f(t)}_{L^1(\mathbb{R}^d)} = \norm{\rho_{f_0}}_{L^1(\mathbb{R}^d)}$ and $\norm{f(t)}_{L^p(\mathbb{R}^d)} = \norm{f_0}_{L^p(\mathbb{R}^d)}$ for all $p \in [1, \infty]$ and $f_0 \geq 0$ implies $f(t) \geq 0$ for all $t \in \mathbb{R}$.
\end{proposition}

\begin{proof}
The conservation of the energy is obtained by a straightforward calculation. The remaining relations hold because $f$ is constant along a Lebesgue measure preserving flow.
\end{proof}

Moreover, we will use the following results stating that the potential energy is dominated by the kinetic energy if $0 < a < 1$ and $f \in L^1 \cap L^{\infty} \left( \mathbb{R}^3 \times \mathbb{R}^3 \right)$.

\begin{lemma}
\label{eq:lemma estimate kinetic versus total energy}
Let $d=3$, $\gamma\in\mathbb{R}$, $0 < a <1$ and $K$ be defined as in \eqref{eq:definition potential}. Moreover, let $f \coloneqq f(x,v) \geq 0$ such that $f \in L^1 \cap L^{\infty} \left( \mathbb{R}^3 \times \mathbb{R}^3 \right)$ and $\mathcal{E}[f]$ be defined as in Proposition \ref{proposition: Conservation laws relativistic Vlasov}. Then there exists a constant $C$ depending only on $a$, $\abs{\gamma}$ and $\norm{f}_{L^1 \cap L^{\infty} (\mathbb{R}^6)}$  such that
\begin{align}
\int_{\mathbb{R}^3 \times \mathbb{R}^3}  \left< v \right> f(x,v)\,dx\,dv - C \leq 2 \mathcal{E}[f]
\leq 3 \int_{\mathbb{R}^3 \times \mathbb{R}^3} \left< v \right> f(x,v)\,dx\,dv + C .
\end{align}
\end{lemma}

\begin{remark}
Note that in the relativistic setting we experience a loss of control in the velocity moments. In fact, it is possible to control only the first moment in velocity by means of the energy of the system, whereas in the non-relativistic setting the kinetic energy controls the second velocity moment. 
\end{remark}

\begin{proof}[Proof of Lemma \ref{eq:lemma estimate kinetic versus total energy}]
By means of the Hardy-Littlewood-Sobolev inequality (see e.g. \cite[Chapter 4.3]{LL2001}) there exists a constant $C >0$ (depending on $a$) such that
\begin{align}
 \int_{\mathbb{R}^3 \times \mathbb{R}^3}  \abs{x-y}^{-a} \rho_f(t,x) \rho_f(t,y)\,dx\,dy 
&\leq C \norm{\rho_f}_{L^p(\mathbb{R}^3)}^2
\quad \text{with} \; \; p = \frac{6}{6-a} .
\end{align} 
Using the second inequality of \eqref{eq: auxiliary estimates electric field 4} we obtain
\begin{align}
  \int_{\mathbb{R}^3 \times \mathbb{R}^3}  \abs{x-y}^{-a} \rho_f(t,x) \rho_f(t,y)\,dx\,dy
&\leq C \norm{f}_{L^{\infty}(\mathbb{R}^{6})}^{\frac{2(p-1)}{p}}
\left( \int_{\mathbb{R}^3 \times \mathbb{R}^3} \abs{v}^{3 (p-1)} f(x,v)\,dx\,dv \right)^{\frac{2}{p}} .
\end{align} 
Note that $3 (p-1) = \frac{3a}{6-a} < \frac{3}{5}$ for all $0 < a <1$. By means of Lemma \ref{lemma:definition and interpolation of velocity moments} we estimate 
\begin{align}
 \int_{\mathbb{R}^3 \times \mathbb{R}^3}   \abs{x-y}^{-a} \rho_f(t,x) \rho_f(t,y)\,dx\,dy 
&\leq C \norm{f}_{L^{\infty}(\mathbb{R}^{6})}^{\frac{2(p-1)}{p}}
\norm{f}_{L^1(\mathbb{R}^{6})}^{\frac{8 - 6 p}{p}}
\left( \int_{\mathbb{R}^3 \times \mathbb{R}^3} \abs{v} f(x,v)\,dx\,dv \right)^{\frac{6(p-1)}{p}}
\nonumber \\
&\leq C \norm{f}_{L^1 \cap L^{\infty}(\mathbb{R}^{6})}^{\frac{6 - 4p}{p}}
\left( \int_{\mathbb{R}^3 \times \mathbb{R}^3} \left< v \right> f(x,v)\,dx\,dv \right)^{\frac{6(p-1)}{p}} .
\end{align}
Since $\frac{6 (p-1)}{p} = a < 1$ the result follows from Young's inequality for products.
\end{proof}

In the next proposition we collect some estimates we will use in Appendix~\ref{section:Well-posedness of the Hartree-Fock equations}.

\begin{proposition}

Let $s \geq 1/2$, $0 < a \leq 1$ and $K(x) = \gamma \frac{1}{\abs{x}^a}$. Then,
\begin{align}
\label{eq:L-infty estimate on convoluted potential}
\norm{K * (fg)}_{L^{\infty}(\mathbb{R}^3)}
&\leq C \abs{\gamma} \norm{f}_{H^{1/2}(\mathbb{R}^3)} \norm{g}_{H^{1/2}(\mathbb{R}^3)}
\quad \text{for all} \quad  f,g \in H^{1/2}(\mathbb{R}^3), 
\\
\label{eq:local well posedness finite rank system estimate for lipschitz property}
\norm{\big( K * (fg) \big) h}_{H^s(\mathbb{R}^3)}
&\leq C \abs{\gamma} \Big[
\norm{f}_{H^s(\mathbb{R}^3)} 
\norm{g}_{H^{\frac{1}{2}}(\mathbb{R}^3)}
\norm{h}_{H^{\frac{1}{2}}(\mathbb{R}^3)}
+
\norm{f}_{H^{\frac{1}{2}}(\mathbb{R}^3)} 
\norm{g}_{H^s(\mathbb{R}^3)}
\norm{h}_{H^{\frac{1}{2}}(\mathbb{R}^3)}
\nonumber \\
&\qquad \quad  + 
\norm{f}_{H^{\frac{1}{2}}(\mathbb{R}^3)} 
\norm{g}_{H^{\frac{1}{2}}(\mathbb{R}^3)}
\norm{h}_{H^s(\mathbb{R}^3)}
\Big]
\quad \text{for all} \quad  f,g,h \in H^s(\mathbb{R}^3).
\end{align}
For $i \in \{1,2\}$, $\omega_{i} \in \textfrak{S}^{1,\frac{1}{2}} \left( L^2(\mathbb{R}^3) \right)$, $\rho_{\omega_i}(x) = N^{-1} \omega_i(x,x)$
and $X_i(x;y) = N^{-1} K(x-y) \omega_i(x;y)$ we have
\begin{align}
\label{eq:Hatree local well-posedness auxiliary estimates 3}
\left|{\int_{\mathbb{R}^3} (K * \rho_{\omega_1})(x) \rho_{\omega_2}(x)\,dx}\right| 
&\leq C \abs{\gamma} N^{-2} \left[ \tr \left( \sqrt{- \Delta} \abs{\omega_1}  \right) \tr \left( \sqrt{- \Delta} \abs{\omega_2}  \right) \right]^{\frac{a}{2}}
\Big( \norm{\omega_1}_{\textfrak{S}^1}  \norm{\omega_2}_{\textfrak{S}^1} \Big)^{\frac{2 - a}{2}} ,
\\
\label{eq:Hatree local well-posedness auxiliary estimates 4}
\abs{\tr \left( K * \rho_{\omega_1} \, \omega_1 \right)}
&\leq C \abs{\gamma} N^{-1} \left( \tr \left( \sqrt{- \Delta} \abs{\omega_1}  \right) \right)^{a}
\Big( \norm{\omega_1}_{\textfrak{S}^1}  \Big)^{2 - a} 
\\
\label{eq:Hatree local well-posedness auxiliary estimates 5}
\abs{\tr \left( X_1 \omega_2 \right) }
&\leq N \left|{\int_{\mathbb{R}^3} (K * \rho_{\abs{\omega_1}})(x) \rho_{\abs{\omega_2}}(x)\,dx}\right| 
= \Big| \tr \left( K * \rho_{\abs{\omega_1}} \abs{\omega_2} \right)  \Big|.
\end{align}
\end{proposition}

\begin{proof}
The first two inequalities are proven by similar estimates as in \cite{Le2007,  H2021}.
If we split the potential into two parts and use Young's inequality, we obtain 
\begin{align}
\norm{K * (fg)}_{L^\infty(\mathbb{R}^3)}
&\leq  \abs{\gamma} \norm{f}_{L^2(\mathbb{R}^3)} \norm{g}_{L^2(\mathbb{R}^3)}
+ \abs{\gamma} \norm{\left( \abs{\cdot}^{-a} \id_{\abs{\cdot} \leq 1} \right) * (fg)}_{L^\infty(\mathbb{R}^3)} .
\end{align}
By means of the Cauchy-Schwarz inequality and the estimate $sup_{y \in \mathbb{R}^3} \int_{\mathbb{R}^3} dx \frac{\abs{u(x)}^2}{\abs{x-y}} \leq C \norm{u}_{H^{1/2}(\mathbb{R}^3)}^2$ (see \cite[inequality (17)]{Le2007})
the second term can be bounded by %\tdc{So $a$ does not enter at all here?}
\begin{align}
\norm{\left( \abs{\cdot}^{-a} \id_{\abs{\cdot} \leq 1} \right) * (fg)}_{L^\infty(\mathbb{R}^3)}
&\leq 
\left(  \sup_{y \in \mathbb{R}^3} \int_{\mathbb{R}^3} dx \, \frac{\abs{f(x)}^2 }{\abs{x-y}}  \right)^{1/2}
\left( \sup_{y \in \mathbb{R}^3} \int_{\mathbb{R}^3} dx \, \frac{\abs{g(x)}^2 }{\abs{x-y}}  \right)^{1/2}
\nonumber \\
&\leq C \norm{f}_{H^{1/2}(\mathbb{R}^3)} \norm{g}_{H^{1/2}(\mathbb{R}^3)} .
\end{align}

Now, let $f,g,h \in H^s(\mathbb{R}^3)$ and $\frac{1}{p} + \frac{1}{q} = \frac{1}{2}$. By the generalized Leibniz rule (see for example \cite[Lemma 5]{Le2007}) we estimate
\begin{align}
\norm{\big( K * (fg) \big) h}_{H^s(\mathbb{R}^3)}
&= \norm{\left( 1 - \Delta \right)^{s/2} \big( K * (fg) \big) h}_{L^2(\mathbb{R}^3)}
\nonumber \\
&\leq \norm{ \big( K * (fg) \big) h}_{L^2(\mathbb{R}^3)}
+ \norm{\left(  - \Delta \right)^{s/2} \big( K * (fg) \big) h}_{L^2(\mathbb{R}^3)}
\nonumber \\
&\leq
\norm{K * (fg)}_{L^{\infty}(\mathbb{R}^3)} \norm{h}_{L^2(\mathbb{R}^3)}
\nonumber \\
&\quad 
+ C \norm{\left(  - \Delta \right)^{s/2}  K * (fg)  }_{L^p(\mathbb{R}^3)} \norm{h}_{L^q(\mathbb{R}^3)}
+ C \norm{ K * (fg)}_{L^{\infty}(\mathbb{R}^3)} \norm{\left(  - \Delta \right)^{s/2} h}_{L^2(\mathbb{R}^3)}
\nonumber \\
&\leq
C \norm{ K * (fg)}_{L^{\infty}(\mathbb{R}^3)} \norm{h}_{H^s(\mathbb{R}^3)}
+ C \norm{\left(  - \Delta \right)^{s/2}  K * (fg)  }_{L^p(\mathbb{R}^3)} \norm{h}_{L^q(\mathbb{R}^3)}  .
\end{align}
The first summand is suitably bounded by means of \eqref{eq:L-infty estimate on convoluted potential}. 
In order to estimate the remaining terms we distinguish between the cases $1/2 \leq s < 3/2$ and $s \geq 3/2$.

\paragraph{The cases $1/2 \leq s < 3/2$:}
We choose $p = \frac{3}{s}$, $q = \frac{6}{3-2s}$ and recall the Sobolev inequality
$\norm{h}_{L^{\frac{6}{3-2s}}(\mathbb{R}^3)}  \leq C \norm{h}_{H^s(\mathbb{R}^3)}$. 
Using the representation of the Riesz potential $\left( - \Delta \right)^{(r-3)/2} (fg) = c_r \abs{\cdot}^{- r} * (fg)$ which holds for $-3 < r < 0$ and some $c_r \in \mathbb{R}$ %(we first assume $fg \in \mathcal{S}(\mathbb{R}^3)$, the space of Schwartz functions, and later extend it by a density argument), 
the fact that $\abs{\cdot}^{- (s+a)} \in L^{\frac{3}{s+a},\infty}(\mathbb{R}^3)$ and the weak Young inequality we estimate
%\tdn{Maybe we can quickly talk about the "density argument" and discuss if we should rewrite the passage about the Riesz potential.}
\begin{align}
\label{eq:general bound for Lipschitz property auxilliary estimate 1}
\norm{\left(  - \Delta \right)^{s/2}  K * (fg)  }_{L^{\frac{3}{s}}(\mathbb{R}^3)}
&= \abs{\gamma} c_a \norm{\left(  - \Delta \right)^{\frac{s+a -3}{2}}  (fg)  }_{L^{\frac{3}{s}}(\mathbb{R}^3)}
\nonumber \\
&\leq C \abs{\gamma}  \norm{\abs{\cdot}^{- (s+a)} *  (fg)  }_{L^{\frac{3}{s}}(\mathbb{R}^3)}
\nonumber \\
&\leq C \abs{\gamma} \norm{f g}_{L^{\frac{3}{3-a}}(\mathbb{R}^3)} .
\end{align}
Together with the Cauchy Schwarz and Sobolev's inequality this shows
\begin{align}
\norm{\left(  - \Delta \right)^{s/2}  K * (fg)  }_{L^{\frac{3}{s}}(\mathbb{R}^3)} \norm{h}_{L^{\frac{6}{3-2s}}(\mathbb{R}^3)}
&\leq C \abs{\gamma} \norm{f}_{L^{\frac{6}{3-a}}(\mathbb{R}^3)}
\norm{g}_{L^{\frac{6}{3-a}}(\mathbb{R}^3)} \norm{h}_{H^s(\mathbb{R}^3)}
\nonumber \\
&\leq C \abs{\gamma} \norm{f}_{H^{1/2}(\mathbb{R}^3)}
\norm{g}_{H^{1/2}(\mathbb{R}^3)}
\norm{h}_{H^s(\mathbb{R}^3)}
\end{align}
for all $1/2 \leq s < 3/2$.

\paragraph{The cases $s \geq 3/2$:}
We choose $p = 6$ and $q=3$. Note that 
$\norm{h}_{L^3(\mathbb{R}^3)} \leq C \norm{h}_{H^{1/2}(\mathbb{R}^3)}$
holds because of Sobolev's inequality. Using again the representation of the Riesz potential and Sobolev's inequality we obtain
\begin{align}
\norm{\left(  - \Delta \right)^{s/2}  K * (fg)  }_{L^6(\mathbb{R}^3)}
&\leq C \abs{\gamma}
\norm{\left(  - \Delta \right)^{ \frac{s+a - 3}{2}}  fg  }_{L^6(\mathbb{R}^3)}
\leq C \abs{\gamma}
\norm{\left(  - \Delta \right)^{ \frac{s+a -2}{2}}  fg  }_{L^2(\mathbb{R}^3)}  .
\end{align}
If $a \in (0, 1]$ such that $s+2 \leq 2$ we can proceed similar as in \eqref{eq:general bound for Lipschitz property auxilliary estimate 1} and obtain
\begin{align}
\norm{\left(  - \Delta \right)^{s/2}  K * (fg)  }_{L^6(\mathbb{R}^3)}
&\leq C \norm{fg}_{L^{\frac{6}{7-2(s+a)}}(\mathbb{R}^3)} .
\end{align}
Since $1 \leq \frac{6}{7- 2(s+a)} \leq 2$ for $\frac{3}{2} \leq s+ a \leq 2$ we have 
\begin{align}
\norm{fg}_{L^{\frac{6}{7-2(s+a)}}(\mathbb{R}^3)}
&\leq C \norm{fg}_{L^1 \cap L^2(\mathbb{R}^3)} \leq C \norm{f}_{L^2 \cap H^1(\mathbb{R}^3)} \norm{f}_{L^2 \cap H^{1/2}(\mathbb{R}^3)}
\leq C \norm{f}_{H^s(\mathbb{R}^3)} \norm{g}_{H^{1/2}(\mathbb{R}^3)}
\end{align}
by interpolation, H\"older's inequality and Sobolev's inequality.
Using the general Leibniz rule we get for $\mu \geq 0$
\begin{align}
\norm{\left( - \Delta \right)^{\mu/2} (fg)}_{L^2(\mathbb{R}^3)}
&\leq C \norm{\left( - \Delta \right)^{\mu/2} f}_{L^6(\mathbb{R}^3)} \norm{g}_{L^3(\mathbb{R}^3)}
+ C \norm{f}_{L^3(\mathbb{R}^3)} \norm{\left( - \Delta \right)^{\mu/2} g}_{L^6(\mathbb{R}^3)}
\nonumber \\
&\leq C \norm{\left( - \Delta \right)^{\frac{\mu + 1}{2}} f}_{L^2(\mathbb{R}^3)} \norm{g}_{H^{1/2}(\mathbb{R}^3)}
+ C \norm{f}_{H^{1/2}(\mathbb{R}^3)} \norm{\left( - \Delta \right)^{\frac{\mu + 1}{2}} g}_{L^2(\mathbb{R}^3)}.
\end{align}
For $s \geq 3/2$ and $0 < a  \leq 1$ such that $s+a \geq 2$ we consequently have
\begin{align}
\norm{\left(  - \Delta \right)^{s/2}  K * (fg)  }_{L^6(\mathbb{R}^3)}
&\leq C \abs{\gamma} \norm{f}_{H^{a+s-1}(\mathbb{R}^3)} \norm{g}_{H^{1/2}(\mathbb{R}^3)}
+ C \abs{\gamma} \norm{f}_{H^{1/2}(\mathbb{R}^3)} \norm{g}_{H^{a+s-1}(\mathbb{R}^3)}
\nonumber \\
&\leq C \abs{\gamma} \norm{f}_{H^{s}(\mathbb{R}^3)} \norm{g}_{H^{1/2}(\mathbb{R}^3)}
+ C \abs{\gamma} \norm{f}_{H^{1/2}(\mathbb{R}^3)} \norm{g}_{H^{s}(\mathbb{R}^3)}.
\end{align}
In total this shows
\begin{align}
\norm{\left(  - \Delta \right)^{s/2}  K * (fg)  }_{L^6(\mathbb{R}^3)} \norm{h}_{L^3(\mathbb{R}^3)}
&\leq 
C \abs{\gamma} \Big[ \norm{f}_{H^{s}(\mathbb{R}^3)} \norm{g}_{H^{1/2}(\mathbb{R}^3)}
+ \norm{f}_{H^{1/2}(\mathbb{R}^3)} \norm{g}_{H^{s}(\mathbb{R}^3)}
\Big] \norm{h}_{H^{1/2}(\mathbb{R}^3)}
\end{align}
for all $s \geq 3/2$.
Summing up, this shows \eqref{eq:local well posedness finite rank system estimate for lipschitz property}.

Next, we continue with \eqref{eq:Hatree local well-posedness auxiliary estimates 3}. Similarly as in the proof of Lemma \ref{eq:lemma estimate kinetic versus total energy} we use the Hardy-Littlewood-Sobolev inequality (see e.g. \cite[Chapter 4.3]{LL2001})  to estimate 
\begin{align}
\left|{\int_{\mathbb{R}^3}  (K* \rho_1)(x) \, \rho_2(x)\,dx}\right| 
&\leq C|\gamma| \norm{\rho_1}_{L^{\frac{6}{6-a}}(\mathbb{R}^3)} \norm{\rho_2}_{L^{\frac{6}{6-a}}(\mathbb{R}^3)},
\end{align}
where $C$ is a numerical constant depending only on $a$.
Let $\{ \lambda_j, \varphi_j \}_{j \in \mathbb{N}}$ be the spectral set of $\omega_1$. Then
\begin{align}
\norm{\rho_1}_{L^{\frac{6}{6-a}} (\mathbb{R}^3)}
&\leq N^{-1} \sum_{j \in \mathbb{N}} \abs{\lambda_j} \norm{\abs{\varphi_j}^2}_{L^{\frac{6}{6-a}}(\mathbb{R}^3)}
\end{align}
with $1 < \frac{6}{6-a} \leq \frac{6}{5}$. By the interpolation inequality we obtain
$\norm{\abs{\varphi_j}^2}_{L^{\frac{6}{6-a}}(\mathbb{R}^3)} 
\leq \norm{\varphi_j}_{L^2(\mathbb{R}^3)}^{2 - a} 
\norm{\varphi_j}_{L^{3}(\mathbb{R}^3)}^{a}$.
Using the Sobolev inequality for higher order fractional derivatives (\cite[Theorem 1.1]{CT2004})
\begin{align}
\label{eq: fractional Sobolev inequality}
\norm{f}_{L^{q}(\mathbb{R}^3)} \leq C \norm{\left( - \Delta \right)^{s/2} f}_{L^2(\mathbb{R}^3)}
\quad \text{where} \; s < \frac{3}{2}, q =\frac{6}{3-2s} \; \text{and} \; f \in H^s(\mathbb{R}^3)
\end{align}
together with H\"older's inequality ($p = \frac{2}{a}$ and $q = \frac{2}{2-a}$) we get
\begin{align}
\norm{\rho_1}_{L^{\frac{6}{6-a}} (\mathbb{R}^3)}
&\leq C N^{-1} 
\left( \sum_{j \in \mathbb{N}} \abs{\lambda_j} 
\norm{\left( - \Delta \right)^{1/4} \varphi_j}_{L^2(\mathbb{R}^3)}^2
\right)^{\frac{a}{2}}
\left( \sum_{j \in \mathbb{N}} \abs{\lambda_j} 
\norm{\varphi_j}_{L^2(\mathbb{R}^3)}^2
\right)^{\frac{2 - a}{2}}  
\nonumber \\
&\leq C N^{-1}
\left( \tr \left( \sqrt{- \Delta} \abs{\omega_1} \right) \right)^{\frac{a}{2}}
\left( \norm{\omega_1}_{\textfrak{S}^1} \right)^{\frac{2-a}{2}} .
\end{align} 
Plugging this expression into the estimate of the potential shows \eqref{eq:Hatree local well-posedness auxiliary estimates 3}. Inequality \eqref{eq:Hatree local well-posedness auxiliary estimates 4} is an immediate consequence of \eqref{eq:Hatree local well-posedness auxiliary estimates 3} and $\tr \left( K * \rho_1 \, \omega_1 \right) = N \int_{\mathbb{R}^3} K * \rho_1(x) \rho_1(x)\,dx$. In order to show \eqref{eq:Hatree local well-posedness auxiliary estimates 5} we use again the spectral decomposition $\{ \lambda_j , \varphi_j \}_{j \in \mathbb{N}}$ of $\omega_1$ and the Cauchy-Schwarz inequality to estimate
\begin{align}
\abs{\omega_1(x;y)} 
&\leq \left( \sum_{j \in \mathbb{N}} \abs{\lambda_j} \abs{\varphi_j(x)}^2 \right)^{1/2}
\left( \sum_{j \in \mathbb{N}} \abs{\lambda_j} \abs{\varphi_j(y)}^2 \right)^{1/2}
\leq \sqrt{\abs{\omega_1}(x;x)} \; \sqrt{\abs{\omega_1}(y;y)} .
\end{align}
Together with the analogue estimate for $\omega_2$ we obtain
\begin{align}
\abs{\tr \left( X_1 \omega_2 \right)}
&\leq N^{-1} \int_{\mathbb{R}^3 \times \mathbb{R}^3} \abs{K(x-y)} 
\Big( \abs{\omega_1}(x;x) \, \abs{\omega_1}(y;y) \, \abs{\omega_2}(x;x) \, \abs{\omega_2}(y;y) \Big)^{1/2}
\nonumber \\
&\leq  N^{-1} \int_{\mathbb{R}^3 \times \mathbb{R}^3} \abs{K(x-y)} 
\abs{\omega_1}(x;x) \, \abs{\omega_2}(y;y)
\nonumber \\
&= N^{-1}  \abs{\int_{\mathbb{R}^3 \times \mathbb{R}^3} K(x-y)
\abs{\omega_1}(x;x) \, \abs{\omega_2}(y;y)} .
\end{align}

\end{proof}

\section{Proofs of the results}\label{section:proofs}

\subsection{Derivation of the relativistic Vlasov equation}
\label{subsection:Derivation of the relativistic Vlasov equation}

In this subsection we provide the proof of Theorem \ref{theorem:derivation relativistic Vlasov from Hartree} on the accuracy of the use of the relativistic Vlasov equation as effective equation for the semi-relativistic Hartree equation. In order to achieve this, we use some estimates proved in \cite{DRS2017} and estimates for singular interaction potentials shown in \cite{LS2020}. For the sake of completeness, we report below the statements of the results we will use.\\ Note that within this section we will use both notations $\tilde{\rho}_t$ and $\tilde{\rho}(t)$ to denote the dependence of $\tilde{\rho}$ on the time variable.

\begin{lemma}[Theorem 4 in \cite{LS2020}]
\label{lemma:LS2020 commutator estimate potential and weyl quantization}
Let $d \in \{2,3\}$, $\omega \in \textfrak{S}^1 \left(L^2(\mathbb{R}^d) \right)$, $a \in \left( \frac{d}{2} - 2 , d- 2 \right]$, $\gamma\in\mathbb{R}$ and $K$ be defined as in \eqref{eq:definition potential}. Let ${\rm b} \coloneqq \frac{d}{a+1}$ so that $\nabla K \in L^{{\rm b}, \infty}(\mathbb{R}^d)$ and ${\rm b}'$ be the conjugated H\"older exponent of ${\rm b}$. Then for any $\mu \in (0, {\rm b}' -1 ]$, there exists a constant $C>0$ such that
\begin{align}
\sup_{z \in \mathbb{R}^d} \norm{\left[ K(\cdot - z) , \omega \right]}_{\textfrak{S}^1} 
&\leq C \abs{\gamma} \norm{\textnormal{diag} \left( \abs{\left[ x , \omega \right]} \right)}_{L^{{\rm b}' - \mu}}^{\frac{1}{2} + \widetilde{\mu}}
\norm{\textnormal{diag} \left( \abs{\left[ x , \omega \right]} \right)}_{L^{{\rm b}' + \mu}}^{\frac{1}{2} - \widetilde{\mu}} ,
\end{align}
for any $\widetilde{\mu} \in \left( 0, \frac{\mu}{2 {\rm b}'} \right)$. 
\end{lemma}

\begin{lemma}[Proposition 3.1 in \cite{LS2020}]
\label{lemma:LS2020 diagonal estimate commutator x and weyl quantization}
Let $d \in \{2,3\}$, $n \in  2 \mathbb{N} $ and define $\sigma = 4 + n$. Then, for any $\widetilde{W}_N \in {W^{3, \infty} (\mathbb{R}^{d} \times \mathbb{R}^d) } \cap H_{\sigma}^{\sigma + 1} (\mathbb{R}^{d} \times \mathbb{R}^d)$ with Weyl quantization $\widetilde{\omega}_N$, there exists a constant $C > 0$ such that 
\begin{align}
\norm{\textnormal{diag} \left( \abs{\left[ x, \widetilde{\omega}_N \right]} \right)}_{L^p(\mathbb{R}^d)}
&\leq C \, \varepsilon N \norm{\nabla_v \widetilde{W}_N}_{W^{2,\infty} (\mathbb{R}^{2d}) \cap H_{\sigma}^{\sigma} (\mathbb{R}^{2d})}
\end{align}
for any $p \in \left[ 1, 1 + \frac{n}{d} \right]$.
\end{lemma}

Moreover, we will need to estimate the operator $C_{N,t} $ with kernel
\begin{align}
\label{eq:derivation Vlasov from Hartree definition of C}
C_{N,t}(x;y)
&= \left[ \left( K * \widetilde{\rho}_t \right)(x) - 
\left( K * \widetilde{\rho}_t \right)(y) 
- \nabla \left( K * \widetilde{\rho}_t \right) \left( \frac{x+y}{2} \right) \cdot (x-y) \right] \widetilde{\omega}_{N,t} (x;y) .
\end{align}
To this end we recall
\begin{lemma}[Proposition 4.4 in \cite{LS2020}]
\label{lemma:derivation Vlasov from Hartree estimae for C-term}
Let $a \in \left( \frac{d}{2} - 2 , d- 2 \right]$, $\gamma\in\mathbb{R}$ and $K$ be defined as in \eqref{eq:definition potential}. Moreover, let $p \in [1,2]$ and $\nu = 4 + a - d$. Then, there exists a constant $C$ independent of $\varepsilon$ such that 
\begin{align}
\norm{C_{N,t}}_{\textfrak{S}^p}
&\leq C \abs{\gamma} \varepsilon^2 N^{\frac{1}{p}} \norm{\widetilde{\rho}(t)}_{L^1 \cap H^{\nu}
(\mathbb{R}^d)} \norm{\nabla_v^2 \widetilde{W}_{N,t}}_{H_{4}^{4} (\mathbb{R}^{2d})} .
\end{align}
\end{lemma}

%\begin{remark}
%\textcolor{red}{In \cite[Proposition 4.4]{LS2020} is written "Under the hypothesis of \cite[Theorem 1]{LS2020} ...". It is not clear to me which assumptions from the main theorem are used except that the right hand side should be finite? Moreover, I do not understand why one needs to define $\nu$ as the positive part since $n + a + 2 - d > 0$ if one uses the parameters as defined above.
%}
%\end{remark}

\begin{proof}[Proof of Theorem \ref{theorem:derivation relativistic Vlasov from Hartree}]

In the following we denote the integral kernel of an operator $O: L^2(\mathbb{R}^{d}) \rightarrow L^2(\mathbb{R}^d)$
by $O(x;y)$ and write the Fourier transform of the kernel as
\begin{align}
\widehat{O}(p;q) = (2 \pi)^{-d} \int_{\mathbb{R}^d \times \mathbb{R}^d}  O(x;y) e^{- i x \cdot p} e^{i y \cdot q}\,dx \, dy .
\end{align}
The Weyl quantization of the solution to the relativistic Vlasov equation $\widetilde{W}_{N,t}$ evolves according to
\begin{align}
i \varepsilon \partial_t \widetilde{\omega}_{N,t} &= A_{N,t} + B_{N,t}
\end{align}
where $A_{N,t}$ and $B_{N,t}$ are defined as
\begin{align}
\widehat{A}_{N,t}(p;q) &= \frac{\varepsilon^2 (p-q) \cdot (p+q)}{2 \sqrt{1 + \frac{\varepsilon^2}{4} (p+q)^2}}
\widehat{\widetilde{\omega}}_{N,t}(p;q)
\end{align}
and
\begin{align}
B_{N,t}(x;y) &= (x-y) \cdot \nabla  \left(K * \widetilde{\rho}_t \right) \left( \frac{x + y}{2} \right) \widetilde{\omega}_{N,t}(x;y) .
\end{align}
The trace norm difference of the solution of the semi-relativistic Hartree equation and the Weyl quantization of the Vlasov solution is bounded by
\begin{align}
\norm{\omega_{N,t} - \widetilde{\omega}_{N,t}}_{\textfrak{S}^1}
&\leq  \norm{\omega_{N} - \widetilde{\omega}_{N}}_{\textfrak{S}^1}
\nonumber \\
\label{eq:derivation Vlasov from Hartree Duhamel 1}
&\quad + \frac{1}{\varepsilon} \int_0^t 
\norm{\left[ \sqrt{1 - \varepsilon^2 \Delta} , \widetilde{\omega}_{N,s} \right] - A_{N,s}}_{\textfrak{S}^1}ds
\\
\label{eq:derivation Vlasov from Hartree Duhamel 2}
&\quad + 
\frac{1}{\varepsilon} \int_0^t 
\norm{\left[ K * \left( \rho_s - \widetilde{\rho}_s \right) , \widetilde{\omega}_{N,s} \right]}_{\textfrak{S}^1}ds 
\\
\label{eq:derivation Vlasov from Hartree Duhamel 3}
&\quad + 
\frac{1}{\varepsilon} \int_0^t  \norm{C_{N,s}}_{\textfrak{S}^1}ds
\end{align} 
with $C_{N,s}$ being defined as in \eqref{eq:derivation Vlasov from Hartree definition of C}.
This estimate is easily derived (see \cite[Section 3]{DRS2017}) by means of a Duhamel expansion.\footnote{To this end it is necessary to show the existence of the two parameter group which is generated by the Hamiltonian 
$h_H(t) = \sqrt{1 - \varepsilon^2 \Delta} + \left( K * \rho_t \right) $. For the singular interactions under consideration this is done in Appendix~\ref{appendix:duhamel}.
} 
In analogy to \cite[Section 3]{DRS2017} (where the case $d= 3$ has been treated) one gets
\begin{align}
\label{eq:derivation of Vlasov bound for the kinetic energy}
\norm{\left( 1 - \varepsilon^2 \Delta \right) \left( 1 + x^2 \right) \left( \left[ \sqrt{1 - \varepsilon^2 \Delta} , \widetilde{\omega}_{N,s} \right] - A_{N,s} \right)}_{\textfrak{S}^2}
&\leq C N^{\frac{1}{2}} \varepsilon^2 
\int_0^t \norm{\widetilde{W}_{N,s}}_{H_3^6(\mathbb{R}^{2d})} \,ds .
\end{align}
H\"older's inequality for Schatten spaces then implies
\begin{align}
\label{eq:derivation of Vlasov bound for the kinetic energy in trace norm}
\eqref{eq:derivation Vlasov from Hartree Duhamel 1}
&\leq  C N \varepsilon
\int_0^t  \norm{\widetilde{W}_{N,s}}_{H_{\sigma}^{\sigma}(\mathbb{R}^{2d})}\,ds 
\end{align}
since
$\norm{\left( 1 - \varepsilon^2 \Delta \right)^{-1} \left( 1 + x^2 \right)^{-1}}_{\textfrak{S}^2} \leq C N^{\frac{1}{2}}$ and $\sigma = 4  + n \geq 6$ holds by assumption.
Note that  
\begin{align}
\label{eq:density-diff}
\norm{\rho_s - \widetilde{\rho}_s}_{L^1(\mathbb{R}^d)} = \sup_{O \in L^{\infty}(\mathbb{R}^d), \norm{O}_{L^{\infty}} \leq 1} 
\Big| \int_{\mathbb{R}^d}  O(x) \left( \rho_s(x) - \widetilde{\rho}_s(x) \right)\,dx\,\Big| \leq \frac{1}{N} \norm{\omega_{N,s} - \widetilde{\omega}_{N,s}}_{\textfrak{S}^1}
\end{align}
because the space of bounded operators is the dual of the space of trace-class operators and every function $x \mapsto O(x)$ defines a multiplication operator. 
Hence
\begin{align}\label{eq:duhamel-trace}
\eqref{eq:derivation Vlasov from Hartree Duhamel 2}
&\leq \frac{1}{\varepsilon} \int_0^t 
\int_{\mathbb{R}^d} \abs{\rho_s(z) - \widetilde{\rho}_s(z)}  \,
\norm{\left[ K (\cdot - z), \widetilde{\omega}_{N,s} \right]}_{\textfrak{S}^1}dz\,ds 
\nonumber \\
&\leq \frac{1}{\varepsilon N} \int_0^t  
\norm{\omega_{N,s} - \widetilde{\omega}_{N,s}}_{\textfrak{S}^1}
\, \sup_{z \in \mathbb{R}^d}  \norm{\left[ K (\cdot - z), \widetilde{\omega}_{N,s} \right]}_{\textfrak{S}^1}\,ds.
\end{align}
Now let ${\rm b}' = \frac{d}{d - (a + 1)}$, $n \in 2 \mathbb{N}$ such that $n > \frac{d}{{\rm b}-1} = d ({\rm b}'-1)$ and $\mu \in (0, {\rm b}' -1 )$ such that ${\rm b}' + \mu \leq 1 + \frac{n}{d}$.
% Moreover, since $1 + \frac{n_1}{d} > 1 + (b' - 1) = b'$, $\nu$ can actually be chosen small enough such that $b' + \nu \leq 1 + \frac{n_1}{d}$.
By means of Lemma \ref{lemma:LS2020 commutator estimate potential and weyl quantization} and
Lemma \ref{lemma:LS2020 diagonal estimate commutator x and weyl quantization} we then get
\begin{align}
\sup_{z \in \mathbb{R}^3} \norm{\left[ K(\cdot - z) , \widetilde{\omega}_{N,s} \right] }_{\textfrak{S}^1} 
&\leq C \abs{\gamma} \norm{\textnormal{diag} \left[ x, \widetilde{\omega}_{N,s} \right] }_{L^{{\rm b}' + \mu} \cap L^{{\rm b}' - \mu}(\mathbb{R}^d)}
\nonumber \\
&\leq C \abs{\gamma} \varepsilon N \norm{\nabla_v \widetilde{W}_{N,s}}_{W^{2, \infty}(\mathbb{R}^{2d}) \cap H_{4 + n}^{4 + n}(\mathbb{R}^{2d})}
\end{align}
and
\begin{align}
\eqref{eq:derivation Vlasov from Hartree Duhamel 2}
&\leq C \abs{\gamma} \int_0^t  
\norm{\omega_{N,s} - \widetilde{\omega}_{N,s}}_{\textfrak{S}^1}
\, 
\norm{\nabla_v \widetilde{W}_{N,s}}_{W^{2, \infty}(\mathbb{R}^{2d}) \cap H_{4 + n}^{4 + n}(\mathbb{R}^{2d})} \, ds .
\end{align}
Together with Lemma \ref{lemma:derivation Vlasov from Hartree estimae for C-term} ($\sigma \geq 6$) this leads to
\begin{align}
\norm{\omega_{N,t} - \widetilde{\omega}_{N,t}}_{\textfrak{S}^1}
&\leq  \norm{\omega_{N} - \widetilde{\omega}_{N}}_{\textfrak{S}^1}
+
\int_0^t  \Big[ \lambda(s) \,
\norm{\omega_{N,s} - \widetilde{\omega}_{N,s}}_{\textfrak{S}^1}
+  N \varepsilon \, C(s) \Big]\,ds ,
\end{align}
where
\begin{align}
\lambda(s) = C \abs{\gamma} \norm{\nabla_v \widetilde{W}_{N,s}}_{W^{2, \infty}(\mathbb{R}^{2d}) \cap H_{4 + n}^{4 + n}(\mathbb{R}^{2d})}
\end{align}
and
\begin{align}
C(s) &=  C \left( 1  + \abs{\gamma} \right) \Big( 1 
+ \norm{\widetilde{\rho}_s}_{L^1(\mathbb{R}^d) \cap H^{\nu}(\mathbb{R}^d)} \Big) \norm{\widetilde{W}_{N,s}}_{H_{\sigma}^{\sigma} (\mathbb{R}^{2d})} .
\end{align}
%\tdn{I am not sure if we should state the dependence of the constant on $d$, $n_1$ and $a$?}
Inequality \eqref{eq:derivation relativistic Vlasov from Hartree}  then follows from Gronwall's Lemma.
\end{proof}

By means of \eqref{eq:density-diff} and Theorem~\ref{theorem:derivation relativistic Vlasov from Hartree} it is possible to control the $L^1$-distance between the semi-relativistic Hartree- and relativistic Vlasov-density. This enables us to extend Theorem~\ref{theorem:derivation relativistic Vlasov from Hartree} to arbitrary $p$-Schatten norms, as shown below.

\begin{proof}[Proof of Proposition~\ref{remark:p-schatten}]
%To generalise the above proof to obtain a convergence in $p$-Schatten norm for $p\in[1,\infty)$, we proceed as follows. The expression to bound is 
We proceed as in the proof of Theorem~\ref{theorem:derivation relativistic Vlasov from Hartree} and replace the trace norm by the $p$-Schatten norm in \eqref{eq:derivation Vlasov from Hartree Duhamel 1}--\eqref{eq:derivation Vlasov from Hartree Duhamel 3}:
\begin{equation}
\label{eq:Duhamel expansion for general Schatten spaces}
\begin{split}
\norm{\omega_{N,t}-\widetilde{\omega}_{N,t}}_{\textfrak{S}^p}
&\leq \norm{\omega_N-\widetilde{\omega}_N}_{\textfrak{S}^p}\\
&+\frac{1}{\e}\int_0^t\norm{\left[\sqrt{1-\e^2\Delta},\widetilde{\omega}_{N,s}\right]-A_{N,s}}_{\textfrak{S}^p} ds\\
&+\frac{1}{N\,\e}\int_0^t  \norm{\omega_{N,s}-\widetilde{\omega}_{N,s}}_{\textfrak{S}^1}\sup_{z \in \mathbb{R}^d} \norm{\left[K(\cdot-z),\widetilde{\omega}_{N,s}\right]}_{\textfrak{S}^p} ds\\
&+\frac{1}{\e}\int_0^t\norm{C_{N,s}}_{\textfrak{S}^p} ds,\\
\end{split}
\end{equation}
where we used \eqref{eq:density-diff} and the analogue of \eqref{eq:duhamel-trace} for $p$-Schatten norms. Notice that the existence of the two parameter semi-group is addressed in Appendix~\ref{appendix:duhamel} for $p$-Schatten norms.\\
We first look at $p\in\left[1, \min \left\{ \frac{d}{a+1} , 2 \right\} \right)$. 
The third term is bounded by \cite[Proposition~4.3]{LS2020}, the refinement of the Calder\' on-Vaillancourt inequality \cite{B1999} and Theorem~\ref{theorem:derivation relativistic Vlasov from Hartree}. The fourth term is bounded in Lemma~\ref{lemma:derivation Vlasov from Hartree estimae for C-term}. Both terms provide a bound of order $O(N^{\frac{1}{p}}\e)$. As for the second term, we proceed by interpolation. By \eqref{eq:derivation of Vlasov bound for the kinetic energy} we get
\begin{equation*}
\norm{\left[\sqrt{1-\e^2\Delta},\widetilde{\omega}_{N,s}\right]-A_{N,s}}_{\textfrak{S}^1}\leq C(t)N\e^2\quad\mbox{and}\quad\norm{\left[\sqrt{1-\e^2\Delta},\widetilde{\omega}_{N,s}\right]-A_{N,s}}_{\textfrak{S}^2}\leq C(t){N}^{\frac{1}{2}}\e^2.
\end{equation*} 
Hence, for $p= \min \left\{ \frac{d}{a+1} , 2 \right\}-\eta$, for $\eta>0$ small,
\begin{equation*}
\norm{{\left[\sqrt{1-\e^2\Delta},\widetilde{\omega}_{N,s}\right]-A_{N,s}}}_{\textfrak{S}^p}\leq C(t)N^{\frac{1}{p}}\e^2.
\end{equation*}
This is enough to get a bound on $\norm{\omega_{N,t}-\widetilde{\omega}_{N,t}}_{\textfrak{S}^p}$ of order $O(N^{\frac{1}{p}}\e)$, $p\in\left[1, \min \left\{ \frac{d}{a+1} , 2 \right\} \right)$.\\
To get a convergence rate for $p\in\left[\min \left\{ \frac{d}{a+1} , 2 \right\},\infty\right)$, we interpolate between $\norm{\omega_{N,t}-\widetilde{\omega}_{N,t}}_{\textfrak{S}^q}$ with $q= \min \left\{ \frac{d}{a+1} , 2 \right\} - \eta$, for $\eta>0$ arbitrarily small, and $\norm{\omega_{N,t}-\widetilde{\omega}_{N,t}}_{\textfrak{S}^\infty}
\leq\norm{\omega_{N,t}}_{\textfrak{S}^{\infty}}+\norm{\widetilde{\omega}_{N,t}}_{\textfrak{S}^{\infty}}\leq C$. This yields 
\begin{equation*}
\norm{\omega_{N,t}-\widetilde{\omega}_{N,t}}_{\textfrak{S}^p}\leq \norm{\omega_{N,t}-\widetilde{\omega}_{N,t}}_{\textfrak{S}^q}^\frac{q}{p}\norm{\omega_{N,t}-\widetilde{\omega}_{N,t}}_{\textfrak{S}^\infty}^{1-\frac{q}{p}}\leq C\norm{\omega_{N,t}-\widetilde{\omega}_{N,t}}_{\textfrak{S}^q}^\frac{q}{p}
\end{equation*}
thus providing a bound of order $O(N^{\frac{1}{p}}\e^\frac{q}{p})$, with $\frac{q}{p}<1$. Keeping track of all the constants and of the dependence on time concludes the proof.
\end{proof}

%\begin{remark}\label{remark:HFterm}
%In the same spirit of the proof of \cite[Theorem~1.3]{LS2020}, we introduce the semigroup $U(t;s)$ defined as the solution to
%\begin{equation*}
% i\varepsilon\partial_t U(t;s)=(\sqrt{1-\varepsilon^2\Delta}+K*\rho_t-X_t)\,U(t;s)
%\end{equation*} 
%Hence, proceeding as in the proof of Theorem~\ref{theorem:derivation relativistic Vlasov from Hartree} we get
%\begin{equation*}
%\begin{split}
%\tr |\omega_{N,t}-\widetilde{\omega}_{N,t}|&\leq \tr |\omega_N-\widetilde{\omega}_N|+\eqref{eq:derivation Vlasov from Hartree Duhamel 1}+\eqref{eq:derivation Vlasov from Hartree Duhamel 2}+\eqref{eq:derivation Vlasov from Hartree Duhamel 3} \\
%&+\frac{1}{\e}\int_0^t \tr |[X_t,(\omega_{N,t}-\widetilde{\omega}_{N,t})]|\,ds
%\end{split}
%\end{equation*}
%Using \cite[Proposition~5.1]{LS2020} and \cite[Proposition~5.2]{LS2020} for the last term, and bounding the remaining terms as in the proof of Theorem~\ref{theorem:derivation relativistic Vlasov from Hartree}, we obtain 
%\begin{equation}
%\tr |\omega_{N,t}-\widetilde{\omega}_{N,t}|\leq\left(\tr |\omega_N-\widetilde{\omega}_N|+C(t)N\e+\widetilde{C}(t)N\e^{{s}-1}\right)\,e^{\lambda(t)}\,,
%\end{equation}
%where $s=\frac{d}{2}-a$, $C(t)$ and $\lambda(t)$ are given in Theorem~\ref{theorem:derivation relativistic Vlasov from Hartree} and $\widetilde{C}\in C^0(\mathbb{R}_+,\mathbb{R}_+)$ depends on $d,\,a$ and the initial data.
%\end{remark}

\subsection{Results about the relativistic Vlasov equation}

\begin{proof}[Proof of Theorem \ref{theorem:velocity moments short time estimate}]
Let $f$ be a solution of \eqref{eq:relativistic Vlasov} and denote by
\begin{align}
\label{eq: definition velocity moments}
M_k(t) = \sup_{0 \leq s \leq t} \int_{\mathbb{R}^d \times \mathbb{R}^d}  \abs{v}^k f(s,x,v)\,dx\,dv  \quad \text{with} \; k \in \mathbb{N}
\end{align}
the velocity moments of solutions of the relativistic Vlasov equation.
%\tdn{I should check again if $f \in L^1 \cap L^{\infty}$ is enough to cancel the boundary terms by the fundamental theorem of calculus.}
Using integration by parts we get
\begin{align}
\frac{d}{dt} M_k(t)
&\leq \Big| \frac{d}{dt} \int_{\mathbb{R}^d \times \mathbb{R}^d}  \abs{v }^k f(t,x,v)\,dx\,dv \Big|
\leq k \int_{\mathbb{R}^d \times \mathbb{R}^d}  \abs{E(t,x)}  \,\abs{v}^{k-1} f(t,x,v)\,dx\,dv.
\end{align}
By means of H\"older's inequality we obtain
\begin{align}
\frac{d}{dt} M_k(t)
&\leq k \left(  \int_{\mathbb{R}^d \times \mathbb{R}^d}   \abs{v}^{k} f(t,x,v)\,dx\,dv \right)^{\frac{k-1}{k}}
\left(  \int_{\mathbb{R}^d \times \mathbb{R}^d}  \abs{E(t,x)}^{k} f(t,x,v)\,dx\,dv \right)^{\frac{1}{k}} ,
\end{align}
leading to
\begin{align}
\frac{d}{dt} M_k^{\frac{1}{k}}(t)
&\leq \left( \int_{\mathbb{R}^d}  \abs{E(t,x)}^k \rho_f(t,x)\,dx \right)^{\frac{1}{k}} .
\end{align}
If we again apply H\"older's inequality with $p= \frac{(k+1) d}{k [d - (a+1)] + d}$ and $q = \frac{(k+1)d}{k (a+1)}$, we get
\begin{align}
\frac{d}{dt} M_k^{\frac{1}{k}}(t)
&\leq \norm{\rho_f(t)}_{L^p(\mathbb{R}^d)}^{\frac{1}{k}} 
\norm{\abs{E(t)}^k}_{L^q(\mathbb{R}^d)}^{\frac{1}{k}}
= \norm{\rho_f(t)}_{L^p(\mathbb{R}^d)}^{\frac{1}{k}} 
\norm{E(t)}_{L^{kq}(\mathbb{R}^d)} .
\end{align}
Using \eqref{eq: auxiliary estimates electric field 4}  we obtain
\begin{align}
\label{eq: propagation of velocity moments intermediate result.}
\frac{d}{dt} M_k^{\frac{1}{k}}(t)
&\leq C \norm{\rho_f(t)}_{L^p(\mathbb{R}^d)}^{\frac{k+1}{k}} 
\leq C \norm{f(t)}_{L^{\infty}(\mathbb{R}^{2d})}^{\frac{(p-1)(k+1)}{p k}}
\left( \int_{\mathbb{R}^d \times \mathbb{R}^d}  \abs{v}^{(p-1)d} f(t,x,v)\,dx\,dv \right)^{\frac{k+1}{p k}} .
\end{align}
For $a \in (-1, d-2]$ such that $k \geq \frac{d a}{d- (a+1)}$ we have $(p-1) d \leq k$.
Thus if we use Lemma \ref{lemma:definition and interpolation of velocity moments} we get
\begin{align}
\frac{d}{dt} M_k^{\frac{1}{k}}(t)
&\leq C \norm{f(t)}_{L^{\infty}(\mathbb{R}^{2d})}^{\frac{a+1}{d}}
\norm{\rho_f(t)}_{L^1(\mathbb{R}^d)}^{\frac{k[d - (a+1)] - d a}{d k}}
M_{k}^{\frac{a+1}{k}}(t) .
\end{align}
Together with $\norm{\rho_f(t)}_{L^1(\mathbb{R}^d)} = \norm{\rho_f(0)}_{L^1(\mathbb{R}^d)}$ and $\norm{f(t)}_{L^{\infty}(\mathbb{R}^{2d})} = \norm{f(0)}_{L^{\infty}(\mathbb{R}^{2d})}$ (see Proposition \ref{proposition: Conservation laws relativistic Vlasov}) and Gronwall's Lemma, this shows the second part of the Theorem and the first part for $a \in \left( -1 , 0 \right]$.
Next, we consider the the case $d=3$ and $0 < a \leq \frac{2k}{3+k}$. For $k \geq \frac{d}{ad + a + 1} = \frac{3}{4 a + 1}$ and $p$ defined as previously we have that $1 \leq 3 (p-1)  \leq k$. 
By Proposition \ref{proposition: Conservation laws relativistic Vlasov} and Lemma \ref{eq:lemma estimate kinetic versus total energy} we have
$M_1(t) \leq 3 M_1(0) + C \left( a, \norm{f(0)}_{L^1 \cap L^{\infty}(\mathbb{R}^3)} \right)$.
If we use Lemma \ref{lemma:definition and interpolation of velocity moments} to estimate the $(p-1)$-th velocity moment on the right-hand side of \eqref{eq: propagation of velocity moments intermediate result.} %\tdc{Something is weird here with the labels of the formulas} 
 by the first and $k$-th velocity moment, we get
\begin{align}
\frac{d}{dt} M_k^{\frac{1}{k}}(t)
&\leq C \norm{f(t)}_{L^{\infty}(\mathbb{R}^{6})}^{\frac{a+1}{3}}
M_1(t)^{\frac{(2-a)k - 3 a}{3(k-1)}}
M_{k}^{\frac{4 k a + k -3}{3 k (k-1)}}(t) 
\nonumber \\
&\leq C \norm{f(0)}_{L^{\infty}(\mathbb{R}^{6})}^{\frac{a+1}{3}}
\left( M_1(0) + C \left( a, \norm{f(0)}_{L^1 \cap L^{\infty}(\mathbb{R}^3)} \right) \right)^{\frac{(2-a)k - 3 a}{3(k-1)}}
M_{k}^{\frac{4 k a + k -3}{3 k (k-1)}}(t).
\end{align}
Now let $0 \leq a \leq \frac{1}{2}$. Then $a \leq \frac{2k}{3+k}$ for all $k \in \mathbb{N}$
and 
$\frac{4ka + k -3}{3 k (k-1)} \leq \frac{1}{k}$. For all $k \geq  \frac{3}{4a +1}$ we consequently obtain  long time estimates of the $k$-th velocity moments by means of Gronwall's Lemma.
For $k \leq \frac{3}{4 a + 1}$ we have $3(p-1) \leq 1$ and the right-hand side of \eqref{eq: propagation of velocity moments intermediate result.} can solely be controlled by $M_1(t) \leq 3 M_1(0) + C \left( a, \norm{f(0)}_{L^1 \cap L^{\infty}(\mathbb{R}^3)} \right)$.
\end{proof}

\begin{proof}[Proof of Corollary \ref{cor:estimates density}]

Using \eqref{eq: auxiliary estimates electric field 1}, the interpolation inequality and Proposition \ref{prop:estimate velocity moments} we obtain  for all $k > \frac{d (a+1)}{d - (a+1)}$ that
\begin{align}
\norm{E(t)}_{L^{\infty}(\mathbb{R}^d)}  &
\leq C \left[
\norm{\rho_f(t)}_{L^1(\mathbb{R}^d)} + \norm{f(t)}_{L^{\infty}(\mathbb{R}^{2d})} 
+ \int_{\mathbb{R}^d \times \mathbb{R}^d} \abs{v}^{k} f(t,x,v)\,dx\,dv 
\right] .
\end{align}
By Theorem \ref{theorem:velocity moments short time estimate} this shows $E \in L^{\infty} \left( [0,T] , L^{\infty}(\mathbb{R}^d) \right)$ because $\norm{\rho_f(t)}_{L^1(\mathbb{R}^d)} = \norm{f_0}_{L^{1}(\mathbb{R}^{2d})}$ and $\norm{f(t)}_{L^{\infty}(\mathbb{R}^{2d})} = \norm{f_0}_{L^{\infty}(\mathbb{R}^{2d})}$.
For $n > d$ we obtain analogously to the non-relativistic case (see \cite[Corollary 5.1]{L2019}) the estimate 
\begin{align}
\norm{\rho_f(t)}_{L^{\infty}(\mathbb{R}^d)}
&\leq C \left( \norm{f(t)}_{L^{\infty}(\mathbb{R}^{2d})}^{\frac{1}{n}} 
+ \norm{f_0 \abs{v}^n}_{L^{\infty}(\mathbb{R}^{2d})}^{\frac{1}{n}} 
+ \int_0^t \norm{f(s)}_{L^{\infty}(\mathbb{R}^{2d})}^{\frac{1}{n}}
\norm{E(s)}_{L^{\infty}(\mathbb{R}^d)}\,ds \right)^n ,
\end{align}
showing the claim.
\end{proof}

\begin{proof}[Proof of Proposition \ref{prop:regularity solution Vlasov}] 
Proposition \ref{prop:regularity solution Vlasov} is proved in analogy to \cite[Proposition A.1]{LS2020}. More explicitly, we show the inequalities
\begin{align}
\label{eq:regularity of Vlasov estimate for electric field 1,infty}
\norm{E(t)}_{W_0^{1, \infty}(\mathbb{R}^d)}
&\leq C J(t)
\left( n \left< t \right> + \ln \left( 1 + \norm{f_0}_{W_n^{1,\infty}(\mathbb{R}^{2d})} \right)  \right)
e^{C \int_0^t J(s)\,ds} ,
\\
\label{eq:regularity estimate explicite dependence}
\norm{f(t)}_{W_n^{\sigma, \infty}(\mathbb{R}^{2d})} 
&\leq 
\left( 1 + \norm{f_0}_{W_n^{\sigma, \infty}(\mathbb{R}^{2d})}^{4^{\sigma}}
\right) e^{C n \big( \left< t \right> + \int_0^t  \norm{E(s)}_{W_0^{1, \infty}(\mathbb{R}^d)}\,ds \big)} ,
\\
\label{eq:regularity estimate explicite dependence H-sigma}
\norm{f(t)}_{H_n^{\sigma}(\mathbb{R}^{2d})}
&\leq \left( 1 + {\sup_{s \in [0,t]} \norm{f(s)}_{W_{n}^{\sigma, \infty}(\mathbb{R}^{2d})}^{2 \sigma}
} \right) 
\norm{f_0}_{H_n^{\sigma}(\mathbb{R}^{2d})} e^{C n \big( \left< t \right> + \int_0^t  \norm{E(s)}_{W_0^{1, \infty}(\mathbb{R}^d)}\,ds \big)},
\end{align}
where $J(t) = 1 + \norm{\rho_f(t)}_{L^1(\mathbb{R}^d)} + \norm{\rho_f(t)}_{L^{\infty}(\mathbb{R}^d)}$.
These together with $\norm{\nabla^{\sigma} \rho_f}_{L^{\infty}(\mathbb{R}^d)} \leq C \norm{f}_{W_n^{\sigma, \infty}(\mathbb{R}^{2d})}$ (which holds because $n >d$) prove the claim. \\ Define the transport operator $T = \frac{v}{\sqrt{1 + v^2}} \cdot \nabla_x  +  E \cdot \nabla_{v} $ and to note that every solution $f(t)$ of \eqref{eq:relativistic Vlasov} satisfies
\begin{equation*}
\partial_t D_z^{\alpha} f(t)  + T \left( D_z^{\alpha} f(t) \right) = - \left[ D_z^{\alpha} , T \right] f(t),
\end{equation*}
for a certain multi-index $\alpha$.
For sufficient regular functions $f,g: \mathbb{R}^{2d} \rightarrow \mathbb{R}$, by means of 
\begin{align}
p \int_{\mathbb{R}^d \times \mathbb{R}^d}
g \abs{f}^{p-2} f T \left( f \right)\,dx\,dv
&= \int_{\mathbb{R}^d \times \mathbb{R}^d} g T \left( \abs{f}^p \right)\,dx\,dv
= - \int_{\mathbb{R}^d \times \mathbb{R}^d} T \left( g \right) \abs{f}^p\,dx\,dv
\end{align}
and
$
T \left( \left< z \right>^{np} \right) \leq np \left< z \right>^{np -1} \left( 1 + \abs{E(t,x)} \right)
$ with $z=(x,v)\in\mathbb{R}^{2d}$,
we estimate 
\begin{align}
\label{eq:Vlasov propagation of regularity time derivative}
&\frac{d}{dt} \norm{\left< z \right>^n D_z^{\alpha} f(t)}_{L^p(\mathbb{R}^{2d})}^p
\nonumber \\
&\quad =  \int_{\mathbb{R}^{2d}}
\abs{D_z^\alpha f}^p T \left( \left< z \right>^{np} \right)\,dz
- p \int_{\mathbb{R}^{2d}}
\left< z \right> ^{np} \abs{D_z^{\alpha} f}^{p-2} 
\left( D_z^{\alpha} f \right)  
\left[ D_z^{\alpha} , T \right] f\,dz 
\nonumber \\
&\quad \leq np \big( 1 + \norm{E(t)}_{L^{\infty}(\mathbb{R}^d)} \big) \norm{\left< z \right>^n D_z^{\alpha} f(t)}_{L^p(\mathbb{R}^{2d})}^p
 + p \int_{\mathbb{R}^{2d}}
\left< z \right> ^{np} \abs{D_z^{\alpha} f}^{p-1} 
\abs{\left[ D_z^{\alpha} , T \right] f }\,dz,
\end{align}
where we omitted the dependence of $f$ on $t,\,x,\,v$.

\paragraph{Inequalities \eqref{eq:regularity of Vlasov estimate for electric field 1,infty} and \eqref{eq:regularity estimate explicite dependence} for $\sigma =1$:}
Using that
\begin{align}
\abs{\left[ D_x^{\alpha} D_v^{\beta} , T \right] f(t,x,v)}
&\leq C \Big( \abs{\nabla_x f(t,x,v)} + \max_{i \in \llbracket 1, d \rrbracket} \abs{\partial_{x_i} E(t,x)}  \abs{\nabla_v f(t,x,v)} \Big)
\end{align}
holds for all $\alpha, \beta \in \mathbb{N}_0^d$ such that $\abs{\alpha} + \abs{\beta} \leq 1$
and the multiplicative Young inequality, $p a b^{p-1} \leq a^p + (p-1) b^p$, we get
\begin{align}
\frac{d}{dt}  \sum_{\abs{\alpha} \leq 1} \norm{\left< z \right>^n D_z^{\alpha} f(t)}_{L^p(\mathbb{R}^{2d})}^p
&\leq 
C np \big( 1 + \norm{E(t)}_{W_0^{1,\infty}(\mathbb{R}^d)} \big)
\sum_{\abs{\alpha} \leq 1}
\norm{\left< z \right>^n  D_z^{\alpha} f(t)}_{L^p(\mathbb{R}^{2d})}^p .
\end{align}
By Gronwall's Lemma we obtain
\begin{align}
\label{eq:Vlasov propagation of regularity intermediate 1}
\norm{f(t)}_{W_n^{1, p}(\mathbb{R}^{2d})}
&\leq  \norm{f_0}_{W_n^{1, p}(\mathbb{R}^{2d})} e^{C n \left( t + \int_0^t  \norm{E(s)}_{W_0^{1,\infty}(\mathbb{R}^d)}\,ds \right)}
\end{align}
and
\begin{align}
\label{eq:gradient moments estiamte for W-n-1-infty}
\norm{f(t)}_{W_n^{1, \infty}(\mathbb{R}^{2d})}
&\leq  \norm{f_0}_{W_n^{1, \infty}(\mathbb{R}^{2d})} 
e^{C n \left( t + \int_0^t  \norm{E(s)}_{W_0^{1,\infty}(\mathbb{R}^d)}\,ds \right)}
\end{align}
if we take the limit $p \rightarrow \infty$. Note that 
\begin{align}
\norm{E(t)}_{W_0^{1, \infty}(\mathbb{R}^d)}
&\leq C \left( 1 + \norm{\rho_f(t)}_{L^1(\mathbb{R}^d)}
+ \norm{\rho_f(t)}_{L^{\infty}(
\mathbb{R}^d)}
\left( 1 + \ln \left( 1 + \norm{\nabla \rho_f(t)}_{L^{\infty}(\mathbb{R}^d)} \right) \right) \right)
\end{align}
because of Lemma \ref{lemma:auxiliary estimates} and 
$\norm{\nabla \rho_f(t)}_{L^{\infty}(\mathbb{R}^d)} \leq C \norm{f(t)}_{W_n^{1, \infty}(\mathbb{R}^{2d})}$ since $n >d$. 
Altogether this gives
\begin{align}
\norm{E(t)}_{W_0^{1,\infty}(\mathbb{R}^d)}
&\leq C \left( 1 + \norm{\rho_f(t)}_{L^1(\mathbb{R}^d)} + \norm{\rho_f(t)}_{L^{\infty}(\mathbb{R}^d)} \right)
\nonumber \\
&\quad \times 
\left(  n \left< t \right> + \ln \left( 1 + \norm{f_0}_{W_n^{1, \infty}(\mathbb{R}^{2d})} \right)  + \int_0^t  \norm{E(s)}_{W_0^{1,\infty}(\mathbb{R}^d)}\,ds  \right) . 
\end{align}
Applying Gronwall's Lemma again leads to \eqref{eq:regularity of Vlasov estimate for electric field 1,infty}.

\paragraph{Inequality \eqref{eq:regularity estimate explicite dependence} for $\sigma > 1$:}
Next, we show
\begin{align}
\label{eq:gradient moments estimate for first proof of induction}
{\sup_{s \in [0,t]}\norm{f(s)}_{W_n^{\sigma, \infty}(\mathbb{R}^{2d})} }
&\leq 
\left( 1 + \norm{f_0}_{W_n^{\sigma, \infty}(\mathbb{R}^{2d})}^2 + {\sup_{s \in [0,t]} \norm{f(s)}_{W_n^{\sigma - 1, \infty}(\mathbb{R}^{2d})}^4 }
\right) e^{C n \left( \left< t \right> + \int_0^t  \norm{E(s)}_{W_0^{1, \infty}(\mathbb{R}^d)}\,ds \right)} ,
\end{align}
%\tdn{To my understanding one needs the supremum for the Gronwall estimate at the end of this paragraph.}
which in combination with \eqref{eq:gradient moments estiamte for W-n-1-infty} implies, by induction, inequality \eqref{eq:regularity estimate explicite dependence} for all $\sigma \in \mathbb{N}$.
Using \eqref{eq:Vlasov propagation of regularity time derivative}, the multiplicative Young inequality and that $\left|D_v^{\gamma} \frac{v}{\left< v \right>}\right| \leq C$ holds for all $\gamma \in \mathbb{N}_0^d$, we obtain
\begin{align}
&\frac{d}{dt} \norm{f(t)}_{W_n^{\sigma, p}(\mathbb{R}^{2d})}^p
\nonumber \\
&\leq n p \left( 1 + \norm{E(t)}_{L^{\infty}(\mathbb{R}^d)} \right)
\norm{f(t)}_{W_n^{\sigma, p}(\mathbb{R}^{2d})}^p
\nonumber \\
&\ + p \sum_{\abs{\beta} + \abs{\gamma} \leq \sigma} \int_{\mathbb{R}^{2d}} \left< z \right>^{np} 
\abs{D_x^{\beta} D_v^{\gamma} f(t)}^{p-1}
\left( 
\left|{\left[ D_v^{\gamma} , \frac{v}{\left< v \right>} \right] \cdot \nabla_x D_x^{\beta} f(t)}\right|
+
\abs{\left[ D_x^{\beta} , E(t,x) \right] \cdot \nabla_v D_v^{\gamma} f(t)}
\right)\,dx\,dv
\nonumber \\
&\leq n p \left( 1 + \norm{E(t)}_{L^{\infty}(\mathbb{R}^d)} \right)
\norm{f(t)}_{W_n^{\sigma, p}(\mathbb{R}^{2d})}^p
\nonumber \\
&\ + p  \sum_{\abs{\beta} + \abs{\gamma} \leq \sigma} \int_{\mathbb{R}^{2d}} \left< z \right>^{np} 
\left( 
\left|{\left[ D_v^{\gamma} , \frac{v}{\left< v \right>} \right]}\right|^p \abs{\nabla_x D_x^{\beta} f(t)}^p
+
\abs{D_x^{\beta} D_v^{\gamma} f(t)}^{p-1}
\abs{\left[ D_x^{\beta} , E(t,x) \right] \cdot \nabla_v D_v^{\gamma} f(t)}
\right)\,dx\,dv
\nonumber \\
&\leq n p \left( 1 + \norm{E(t)}_{L^{\infty}(\mathbb{R}^d)} \right)
\norm{f(t)}_{W_n^{\sigma, p}(\mathbb{R}^{2d})}^p
\nonumber \\
&\ + p \sum_{\abs{\beta} + \abs{\gamma} \leq \sigma} \int_{\mathbb{R}^{2d}} \left< z \right>^{np} 
\abs{D_x^{\beta} D_v^{\gamma} f(t)}^{p-1}
\abs{\left[ D_x^{\beta} , E(t,x) \right] \cdot \nabla_v D_v^{\gamma} f(t)}\,dx\,dv,
\end{align}
where we omitted the dependence of $f$ from $x$ and $v$ and use the variable $z=(x,v)\in\mathbb{R}^{2d}$.
Note that, for $\beta$, $\gamma$ and $\delta$ multi-indices, 
\begin{align}
&\sum_{\abs{\beta} + \abs{\gamma} \leq \sigma}  
\abs{D_x^{\beta} D_v^{\gamma} f(t)}^{p-1}
\abs{\left[ D_x^{\beta} , E(t) \right] \cdot \nabla_v D_v^{\gamma} f(t)}
\nonumber \\
&\quad \leq \norm{E(t)}_{W_0^{1,\infty}(\mathbb{R}^d)}
\sum_{\abs{\beta} + \abs{\gamma} \leq \sigma} \sum_{\abs{\delta} = \abs{\beta} - 1}
\abs{D_x^{\beta} D_v^{\gamma} f(t)}^{p-1}
\abs{D_x^{\delta} \nabla_v D_v^{\gamma} f(t)}
\nonumber \\
&\qquad 
+ 
\sum_{\abs{\beta} + \abs{\gamma} \leq \sigma} \sum_{\abs{\delta} \leq \abs{\beta} - 2}
\abs{D_x^{\beta} D_v^{\gamma} f(t)}^{p-1}
\norm{E(t)}_{W_0^{\sigma,\infty}(\mathbb{R}^d)}
\abs{D_x^{\delta} \nabla_v D_v^{\gamma} f(t)} 
\nonumber \\
&\quad \leq
C \left( 1 + \norm{E(t)}_{W_0^{1,\infty}(\mathbb{R}^d)} \right)
\sum_{\abs{\beta} + \abs{\gamma} \leq \sigma}  \abs{D_x^{\beta} D_v^{\gamma} f(t)}^{p}
+ \frac{C}{p}  \norm{E(t)}_{W_0^{\sigma,\infty}(\mathbb{R}^d)}^p
\sum_{\abs{\beta} + \abs{\gamma} \leq \sigma - 1}  \abs{D_x^{\beta} D_v^{\gamma} f(t)}^{p},
\end{align}
%\tdn{Here, is a factor of $\sigma$ missing and it should rather be $C \left( \sigma + \norm{E(t)}_{W_0^{1,\infty}(\mathbb{R}^d)} \right) ...$. However, this is not important for the discussion since we are not taking the limit $\sigma \rightarrow \infty$ and I would therefore include it in the constant $C$.}
and therefore
\begin{align}
\frac{d}{dt} \norm{f(t)}_{W_n^{\sigma, p}(\mathbb{R}^{2d})}^p
&\leq C n p \left( 1 + \norm{E(t)}_{W_0^{1, \infty}(\mathbb{R}^d)} \right)
\norm{f(t)}_{W_n^{\sigma, p}(\mathbb{R}^{2d})}^p
+ C
\norm{E(t)}_{W_0^{\sigma,\infty}(\mathbb{R}^d)}^p
\norm{f(t)}_{W_n^{\sigma - 1, p}(\mathbb{R}^{2d})}^p.
\end{align}
Applying Gronwall's Lemma, taking the $p$-th root and using \eqref{eq: auxiliary estimates electric field 4} lead to
\begin{align}
\label{eq:gradient moments estimate for first proof of induction intermediate}
\norm{f(t)}_{W_n^{\sigma,p}(\mathbb{R}^{2d})}
&\leq e^{C n  \big( \left< t \right> + \int_0^t  \norm{E(s)}_{W_0^{1, \infty}(\mathbb{R}^d)}\,ds \big)}
\left( 
\norm{f_0}_{W_n^{\sigma,p}(\mathbb{R}^{2d})}
+
{\sup_{s \in [0,t]}\norm{E(s)}_{W_0^{\sigma, \infty}(\mathbb{R}^d)} \norm{f(s)}_{W_{n}^{\sigma - 1, p}(\mathbb{R}^{2d})} }
\right)
\nonumber \\
&\leq e^{C n  \big( \left< t \right> + \int_0^t  \norm{E(s)}_{W_0^{1, \infty}(\mathbb{R}^d)}\,ds \big)}
\Big[ 
\norm{f_0}_{W_n^{\sigma,p}(\mathbb{R}^{2d})}
\nonumber \\
&\qquad 
+
{\sup_{s \in [0,t]}\norm{f(s)}_{W_{n}^{\sigma - 1, p}(\mathbb{R}^{2d})}
\left( 1 +  \norm{f(s)}_{W_n^{\sigma -1, \infty}(\mathbb{R}^{2d})} \right)  \left( 1 + \ln \big( 1 + \norm{f(s)}_{W_n^{\sigma, \infty}(\mathbb{R}^{2d})} \big) \right)  \Big]. }
\end{align}
Taking the limit $p \rightarrow \infty$ and estimating the logarithm by $C (1 + {\norm{f(s)}^{1/2}_{{W}_n^{\sigma, \infty}(\mathbb{R}^{2d})} } )$ give
\eqref{eq:gradient moments estimate for first proof of induction}.

\paragraph{Inequality \eqref{eq:regularity estimate explicite dependence H-sigma}: }
Inequality \eqref{eq:Vlasov propagation of regularity intermediate 1} proves \eqref{eq:regularity estimate explicite dependence H-sigma} for $\sigma = 1$.
From \eqref{eq:gradient moments estimate for first proof of induction intermediate} we get
\begin{align}
\norm{f(t)}_{H_n^{\sigma}(\mathbb{R}^{2d})}
&\leq 
e^{C n  \big( \left< t \right> + \int_0^t  \norm{E(s)}_{W_0^{1, \infty}(\mathbb{R}^d)}\,ds \big)}
{\sup_{s \in [0,t]}
\left( 1 + \norm{f(s)}_{W_n^{\sigma, \infty}(\mathbb{R}^{2d})}^2 \right)
\left(
\norm{f_0}_{H_n^{\sigma}(\mathbb{R}^{2d})}
+ \norm{f(s)}_{H_n^{\sigma - 1}(\mathbb{R}^{2d})}
\right) }
\end{align}
which, by induction, enables to infer \eqref{eq:regularity estimate explicite dependence H-sigma} for all $\sigma \in \mathbb{N}$.

\end{proof}

\begin{proof}[Proof of Proposition \ref{proposition: uniqueness Vlasov}]
The statement is proven in analogy to \cite[Proposition 2.1]{LS2020}. 
By means of 
\begin{align}
\partial_t (f_1 - f_2 ) = - \left< v \right> \cdot \nabla_x (f_1 - f_2 ) + \nabla K * \rho_{f_1} \cdot \nabla_v ( f_1 - f_2 ) + \nabla K * (\rho_{f_2} - \rho_{f_1} ) \cdot \nabla_v f_2
\end{align}
we denote $f:=f_1-f_2$ and estimate
\begin{align}
\frac{d}{dt} \int_{\mathbb{R}^d \times \mathbb{R}^d}  \abs{(f_1- f_2)(t,x,v)}\,dx\,dv
&\leq 
 \int_{\mathbb{R}^d \times \mathbb{R}^d} \text{sign}[f(t,x,v)] \left( \nabla K * (\rho_{f_1} - \rho_{f_2}) \right)(t,x) \cdot \nabla_v f_2(t,x,v)\,dx\,dv
\nonumber \\
&\leq \int_{\mathbb{R}^d \times \mathbb{R}^d \times \mathbb{R}^d}  \abs{\nabla K(x-y)} (\rho_{f_1} - \rho_{f_2})(t,y) \abs{\nabla_v f_2(t,x,v)}\,dx\,dv\,dy
\nonumber \\
&\leq \norm{\rho_{f_1} - \rho_{f_2}}_{L^1(\mathbb{R}^d)}
\norm{\abs{\nabla K} * \left( \int_{\mathbb{R}^d}  \abs{\nabla_ v f_2(t, \cdot, v)}\,dv \right)}_{L^{\infty}(\mathbb{R}^d)}.
\end{align}
Using \eqref{eq: auxiliary estimates electric field 1} we obtain
\begin{align}
\frac{d}{dt} \norm{(f_1 - f_2)(t)}_{L^1(\mathbb{R}^{2d})}
&\leq C \norm{(f_1 - f_2)(0)}_{L^1(\mathbb{R}^{2d})} \norm{\rho_{\abs{\nabla_v f_2}}(t)}_{L^{\frac{d}{d - (a+1)} + \delta} \cap L^{\frac{d}{d - (a+1)} - \delta} (\mathbb{R}^d)}
\end{align}
for arbitrary $0 < \delta < \frac{d}{d - (a+1)}$. Together with Gronwall's Lemma this leads to \eqref{eq:uniqueness Vlasov}.
\end{proof}

\appendix

\section{Well-posedness of the Hartree and Hartree-Fock equations}
\label{section:Well-posedness of the Hartree-Fock equations}

In this section we will prove Proposition~\ref{lemma: Hartree equation global existence of solution}. We restrict our consideration to the Hartree-Fock equation \eqref{eq:Hartree-Fock equation relativistic}. By setting the exchange term to zero in each estimate one obtains the respective result for the Hartree equation \eqref{eq:Hartree equation relativistic}.
We first consider a finite rank version of equation \eqref{eq:Hartree-Fock equation relativistic} and sketch how one generalizes the global well-posedness result from \cite{FL2007} to general inverse power law potentials. Following the approach of \cite{C1976} this will allow us to finally prove Proposition~\ref{lemma: Hartree equation global existence of solution}. It should be pointed out that this route was used in \cite{AMS2008,HLLS2010} to obtain the result for the case $a=1$. We, nevertheless, decided to present the details for the convenience of the reader.
Throughout this section we will use the notations $D^{s} = \left( 1 - \Delta \right)^{s/2}$ and $D_{\varepsilon}^s = \left( 1 - \varepsilon^2 \Delta \right)^{s/2}$. The norm of the space $\textfrak{S}^{1, \frac{1}{2}}$ can then be written as
\begin{align}
\norm{\omega}_{\textfrak{S}^{1,\frac{1}{2}}}
&= \norm{ D^{1/2} \omega  D^{1/2} }_{\textfrak{S}^1} .
\end{align}

\subsection{Finite rank system}

In the proof of Proposition~\ref{lemma: Hartree equation global existence of solution} we will use that the following set of $M$ coupled equations
\begin{align}
\label{eq:Hartree-Fock equation for orbitals}
i \varepsilon \partial_t \psi_{k}(t) = \sqrt{1 - \varepsilon^2 \Delta} \, \psi_k(t) + \frac{1}{N} \sum_{l=1}^M  \left( K * \abs{\psi_l(t)}^2 \right) \psi_k(t) 
- \frac{1}{N} \sum_{l=1}^M  \left( K * \{ \overline{\psi_l(t)} \psi_k(t) \} \right) \psi_l(t) 
\end{align}
is globally well-posed in $H^{s,M} = \left( H^s(\mathbb{R}^3) \right)^{\times M}$ with the norm
$
\norm{\left\{ \psi_{k} \right\}_{k=1}^M}_{H^{s,M}} = \left( \sum_{k=1}^M \norm{\psi_k}_{H^s(\mathbb{R}^3)}^2 \right)^{1/2} 
$
and $s= 1/2$. For the case $a=1$ this has been proven in \cite{FL2007}. We will rely on the following slight generalization of \cite[Theorem 2.1 and Theorem 2.2]{FL2007}.
\begin{lemma}
\label{lemma: Hartree equation  orbitals global existence}
Let $s \geq 1/2$, $0 < a \leq 1$, $\gamma \in \mathbb{R}$ and $M, N \geq 1$ be integers.
Let $\{ \psi_{k,0} \}_{k=1}^M \subset H^s(\mathbb{R}^3)$ satisfying $0 \leq \scp{\psi_{k,0}}{\psi_{l,0}} \leq \delta_{k,l}$, where $\delta_{k,l}$ is the Kronecker delta. If $\gamma < 0$ and $a = 1$, in addition, assume that 
\begin{align}
\abs{\gamma} < \frac{N \varepsilon}{\gamma_{\rm cr} \left( \sum_{k=1}^M \norm{\psi_{k,0}}^2_{L^2(\mathbb{R}^3)} \right)^{2/3}} ,
\end{align}
where $\gamma_{\rm cr}$ is a universal constant of order $1$. Then, there exists a unique global solution, $\{ \psi_k(t) \}_{k=1}^M \subset H^s(\mathbb{R}^3)$ solving \eqref{eq:Hartree-Fock equation for orbitals} such that 
\begin{align}
\psi_k(t) = \psi_{k,0} \quad \text{and} \quad
\psi_k \in C^0 \left( \mathbb{R}_{+} , H^s(\mathbb{R}^3) \right) \cap C^1 \left( \mathbb{R}_{+} , H^{s-1} (\mathbb{R}^3) \right)
\end{align}
holds, for all $k= 1, \ldots, N$. The solution continuously depends on the initial data and 
\begin{align}
\scp{\psi_k(t)}{\psi_l(t)}_{L^2(\mathbb{R}^3)} = \scp{\psi_{k,0}}{\psi_{l,0}}_{L^2(\mathbb{R}^3)}
\quad \text{and} \quad
\mathcal{E}_{\rm HF} \left[\{ \psi_k(t) \}_{k=1}^M \right]
= 
\mathcal{E}_{\rm HF} \left[\{ \psi_k(0) \}_{k=1}^M \right] 
\end{align}
hold for all $1 \leq k , l \leq M$ and $t \in \mathbb{R}_{+}$ with the energy $\mathcal{E}_{\rm HF}$ being defined as in 
\eqref{eq:definition energy finite rank Hartree system}.
\end{lemma}

\begin{proof}[Proof of Lemma \ref{lemma: Hartree equation  orbitals global existence}]

Let $s \geq 1/2$, $M \geq 1$ be an integer and  $\boldsymbol{F}= (F_1, \ldots, F_M): H^{s,M} \rightarrow H^{s,M}$ be given by
$F_k \left( \{ \psi_k \}_{k=1}^M \right)
= \frac{1}{N} \sum_{l=1}^M  \left( K * \abs{\psi_l(t)}^2 \right) \psi_k(t) 
- \frac{1}{N} \sum_{l=1}^M  \left( K * \{ \overline{\psi_l(t)} \psi_k(t) \} \right) \psi_l(t) $. Using \eqref{eq:local well posedness finite rank system estimate for lipschitz property} and straightforward manipulations we obtain for  
\begin{align}
\norm{\boldsymbol{F}(\{ \psi_k \}_{k=1}^M) - \boldsymbol{F}(\{ \varphi_k \}_{k=1}^M)}_{H^{s,M}}
&\leq C \left( \norm{\{ \psi_k \}_{k=1}^M}^2_{H^{s,M}} + \norm{\{ \varphi_k \}_{k=1}^M}^2_{H^{s,M}} \right)
\norm{\{ \psi_k - \varphi_k \}_{k=1}^M }_{H^{s,M}} ,
\\
\label{eq:well-posedness Hartree Fock nonlinearity auxilliary estimate}
\norm{\boldsymbol{F}(\{ \psi_k \}_{k=1}^M)}_{H^{s,M}}
&\leq C \norm{\{ \psi_k \}_{k=1}^M}^2_{H^{\frac{1}{2},M}} \norm{\{ \psi_k \}_{k=1}^M}_{H^{s,M}} ,
\end{align}
for all $\{ \psi_k \}_{k=1}^M, \{ \varphi_k \}_{k=1}^M \in H^{s,M}$.
%\tdn{The second bound is better than in  \cite{Le2007} and \cite{FL2007}. It is not really clear to me why they are proving a weaker statement because I coudn't find an error in my calculations.}
By standard methods we obtain the local-in-time existence and uniqueness of $\{ \psi_k(t) \}_{k=1}^M$ as well as the continuous dependence on the initial data, i.e. the analogue of \cite[Theorem 2.1]{FL2007}. If $t$ is smaller than the maximal time of existence we, moreover, have both that 
$\scp{\psi_k(t)}{\psi_l(t)}_{L^2(\mathbb{R}^3)} = \scp{\psi_k(0)}{\psi_l(0)}_{L^2(\mathbb{R}^3)}$ for all $1 \leq k , l \leq M$ and that the energy, defined by
\begin{align}
\label{eq:definition energy finite rank Hartree system}
\mathcal{E}_{\rm HF} \left[\{ \psi_k(t) \}_{k=1}^M \right] 
&= \tr \left( \sqrt{ 1 - \varepsilon^2 \Delta} \, \omega_t^{\leq M} \right) + \frac{1}{2} \tr \left( \left( K * \rho_t^{\leq M} - X_t^{\leq M} \right) \, \omega^{\leq M}_t \right)  
\end{align}
with $\omega_t^{\leq M} = \sum_{j=1}^M \ket{\psi_j(t)} \bra{\psi_j(t)}$, $\rho_t^{\leq M}(x) = N^{-1} \omega_t^{\leq M}(x;x)$ and
$X_t^{\leq M}(x;y) = N^{-1} K(x-y) \omega_t^{\leq M}(x;y)$, are conserved quantities.\\
By means of the integral version of \eqref{eq:Hartree-Fock equation for orbitals}, inequality \eqref{eq:well-posedness Hartree Fock nonlinearity auxilliary estimate} and Gronwall's Lemma one can show in the same spirit of \cite[p. 57]{Le2007} that, for any $s > 1/2$ and all times $T_*$ smaller than the maximal time of existence
\begin{align}
\sup_{0 \leq t \leq T_*} \norm{\{ \psi_k(t) \}_{k=1}^M}_{H^{s,M}}
&\leq C \Big(  T_* ,
\norm{\{ \psi_k(0) \}_{k=1}^M}_{H^{s,M}},
\sup_{0 \leq t \leq T_*} \norm{\{ \psi_k(t) \}_{k=1}^M}_{H^{\frac{1}{2},M}}
\Big) .
\end{align}
This implies that the maximal time of existence of any $H^s$-valued solution with $s> 1/2$ is the same as the maximal time of existence of the $H^{1/2}$-valued solution.
Hence it suffices to show global well-posedness of the $H^{1/2}$-valued solution. In this regard note that for $0 < a < 1$
\begin{align}
\Big| \mathcal{E}_{\rm HF} \left[\{ \psi_k(t) \}_{k=1}^M \right] 
- \tr \left( \sqrt{1 - \varepsilon^2 \Delta} \omega_t^{\leq M}\right)
\Big|
&\leq \frac{1}{2} \tr \left( \sqrt{- \varepsilon^2 \Delta} \; \omega_t^{\leq M} \right)
+ C \left( \varepsilon N^{-1} \right)^{\frac{1}{1-a}} \left( \tr \left( \omega_t^{\leq M} \right) \right)^{\frac{2-a}{1-a}}
\end{align}
holds because of  \eqref{eq:Hatree local well-posedness auxiliary estimates 4}, \eqref{eq:Hatree local well-posedness auxiliary estimates 5} and Young's inequality for products. We then get 
\begin{align}
\label{eq: global well-posedness of H-1/2 solutions estimate kinetic energy}
\norm{\{ \psi_k(t) \}_{k=1}^M}_{H^{\frac{1}{2},M}}^2
&\leq  \varepsilon^{-1}  \tr \left( \sqrt{1 - \varepsilon^2 \Delta} \, \omega_t^{\leq M} \right)
\nonumber \\
&\leq 2 \varepsilon^{-1} 
\mathcal{E}_{\rm HF} \left[\{ \psi_k(t) \}_{k=1}^M \right]  + C \left( \varepsilon^a N^{-1} \right)^{\frac{1}{1-a}} 
\left( \tr \left( \omega_t^{\leq M} \right) \right)^{\frac{2-a}{1-a}}
\\
&\leq C \varepsilon^{-1} \tr \left( \sqrt{- \varepsilon^2 \Delta} \, \omega_0^{\leq M} \right)
+ C \left( \varepsilon^a N^{-1} \right)^{\frac{1}{1-a}} 
\left( \tr \left( \omega_0^{\leq M} \right) \right)^{\frac{2-a}{1-a}}
\end{align}
for all $0< a < 1$ and $\gamma \in \mathbb{R}$
by the conservation of the energy and mass.  
If $\gamma \in \mathbb{R}_{+}$ the second summand on the right-hand side of \eqref{eq:definition energy finite rank Hartree system} is positive because of \eqref{eq:Hatree local well-posedness auxiliary estimates 5}. Together with  \eqref{eq:Hatree local well-posedness auxiliary estimates 4} and the conservation of energy this gives
\begin{align}
\norm{\{ \psi_k(t) \}_{k=1}^M}_{H^{\frac{1}{2},M}}^2
&\leq  \varepsilon^{-1} 
\mathcal{E}_{\rm HF} \left[\{ \psi_k(t) \}_{k=1}^M \right]  
\leq 
C \varepsilon^{-1} \tr \left( \sqrt{1 - \varepsilon^2 \Delta} \omega_0^{\leq M} \right) \left( 1 + \tr \left( \omega_0^{\leq M} \right) \right) 
\end{align} 
for $a= 1$ and $\gamma \in \mathbb{R}_{+}$.
For $a= 1$ and $\gamma \in \mathbb{R}_{-}$ we proceed as in \cite{FL2007}. First, we notice that the exchange term is negative in this case. Second, we use the Hardy--Littlewood--Sobolev inequality and the interpolation inequality to estimate
\begin{align}
\tr \left( K * \rho^{\leq M}_t \, \omega^{\leq M}_t \right)
&\leq C \,  N \, \abs{\gamma} \,   \left( \int_{\mathbb{R}^3} \rho_t^{\leq M}(x)\,dx \right)^{2/3} \int_{\mathbb{R}^3} (\rho^{\leq M}_t(x))^{4/3}\,dx
\nonumber \\
\end{align}
Third, we apply $\int_{\mathbb{R}^3}\rho_t^{\leq M}(x)\,dx  = N^{-1} \tr \left( \omega^{\leq M}_0 \right)$ and \cite[Lemma A.1]{FL2007}
to get
\begin{align}
\tr \left( K * \rho_t^{\leq M} \, \omega^{\leq M}_t \right)
&\leq C \, N^{-1} \, \abs{\gamma} \,   \left( \tr \left( \omega^{\leq M}_0 \right) \right)^{2/3}
\norm{\{ \psi_k(t) \}_{k=1}^M}_{H^{\frac{1}{2},M}}^2 .
\end{align}
This leads to
\begin{align}
\norm{\{ \psi_k(t) \}_{k=1}^M}_{H^{\frac{1}{2},M}}^2
&\leq \varepsilon^{-1}  \mathcal{E}_{\rm HF} \left[\{ \psi_k(t) \}_{k=1}^M \right] 
+ (2 \varepsilon )^{-1} \abs{\tr \left( K * \rho_t^{\leq M} \, \omega^{\leq M}_t \right) }
\nonumber \\
&\leq \varepsilon^{-1}  \mathcal{E}_{\rm HF} \left[\{ \psi_k(0) \}_{k=1}^M \right] 
+ \frac{C \abs{\gamma}}{2 N \varepsilon}
\left( \tr \left( \omega^{\leq M}_0 \right) \right)^{2/3} \norm{\{ \psi_k(t) \}_{k=1}^M}_{H^{\frac{1}{2},M}}^2 .
\end{align}
Assuming $1  > \left( \frac{C \abs{\gamma}}{2 N \varepsilon} \tr \left( \omega^{\leq M}_0 \right)^{2/3} \right)$ we obtain
$\norm{\{ \psi_k(t) \}_{k=1}^M}_{H^{\frac{1}{2},M}}^2
\leq C \varepsilon^{-1}  \mathcal{E}_{\rm HF} \left[\{ \psi_k(0) \}_{k=1}^M \right]$.
\end{proof}

\subsection{Proof of Proposition~\ref{lemma: Hartree equation global existence of solution}}

Note that the integrated form of \eqref{eq:Hartree-Fock equation relativistic} is given by
\begin{align}
\label{eq:Hartree equation relativistic integral version}
\omega_{N,t}
&= e^{- i \sqrt{1 - \varepsilon^2 \Delta} \, t} \omega_{N,0}  e^{ i \sqrt{1 - \varepsilon^2 \Delta} \, t}
-  \int_0^t 
e^{- i \sqrt{1 - \varepsilon^2 \Delta} \, (t-s)} \, i
\left[ \left( K * \rho_s  - X_s \right), \omega_{N,s} \right]
e^{ i \sqrt{1 - \varepsilon^2 \Delta} \, (t-s)}\,ds.
\end{align}

\begin{lemma}
\label{lemma: Hartree equation local existence of solution}
For $a \in (0,1]$ equation \eqref{eq:Hartree equation relativistic integral version} has a unique local solution in $\textfrak{S}^{1,\frac{1}{2}} \left( L^2(\mathbb{R}^3) \right)$.
\end{lemma}

\begin{proof}[Proof of Lemma \ref{lemma: Hartree equation local existence of solution}]
Since $\sqrt{1 - \varepsilon^2 \Delta}$ is self-adjoint and commutes with $D ^{\frac{1}{2}}$ we have that $\mathbb{R}_+ \rightarrow \textfrak{S}^{1,\frac{1}{2}} \left( L^2(\mathbb{R}^d) \right)$, $\omega \mapsto e^{- i \sqrt{1 - \varepsilon^2 \Delta} t} \omega e^{ i \sqrt{1 - \varepsilon^2 \Delta} t}$ defines a strongly continuous semigroup. Lemma \ref{lemma: Hartree equation local existence of solution} then follows from
\cite[Theorem 1]{Segal1963} and the Lipschitz property of $\textfrak{S}^{1,\frac{1}{2}} \left( L^2(\mathbb{R}^3) \right) \rightarrow \textfrak{S}^{1,\frac{1}{2}} \left( L^2(\mathbb{R}^3) \right)$, $\omega \mapsto i \left[ \left( K * \rho - X \right) , \omega \right]$. To prove this fact we let
$\omega, \widetilde{\omega} \in \textfrak{S}^{1,\frac{1}{2}}_{+} \left( L^2(\mathbb{R}^3) \right)$, $\rho(x) = N^{-1} \omega(x;x)$ and $X(x;y) = N^{-1} K(x-y) \omega(x;y)$. By H\"older's inequality for Schatten spaces we obtain
\begin{align}
\norm{i \left[ K * \rho , \widetilde{\omega} \right]}_{\textfrak{S}^{1,\frac{1}{2}}}
&\leq 
2 \norm{D^{\frac{1}{2}} K * \rho \, \widetilde{\omega}^{1/2}}_{\textfrak{S}^{2}} \norm{\widetilde{\omega}^{1/2} D^{\frac{1}{2}} }_{\textfrak{S}^{2}}
\nonumber \\
&=
2 \norm{D^{\frac{1}{2}} K * \rho \, \widetilde{\omega} \, K * \rho \, D^{\frac{1}{2}}}_{\textfrak{S}^{1}}^{1/2} \norm{\widetilde{\omega}}_{\textfrak{S}^{1,\frac{1}{2}}}^{1/2}
\nonumber \\
&\leq 
2 \norm{D^{\frac{1}{2}} K * \rho \, D^{-\frac{1}{2}}}_{\textfrak{S}^{\infty}}^{1/2} \norm{ D^{-\frac{1}{2}} \, K * \rho \, D^{\frac{1}{2}}}_{\textfrak{S}^{\infty}}^{1/2} \norm{D^{\frac{1}{2}} \widetilde{\omega} D^{\frac{1}{2}}}_{\textfrak{S}^{1}}^{1/2} \norm{\widetilde{\omega}}_{\textfrak{S}^{1,\frac{1}{2}}}^{1/2}
\nonumber \\
&\leq 2 \norm{D^{\frac{1}{2}} K * \rho \, D^{-\frac{1}{2}}}_{\textfrak{S}^{\infty}}  \norm{\widetilde{\omega}}_{\textfrak{S}^{1,\frac{1}{2}}} .
\end{align}
Similarly,
$\norm{i \left[ X , \widetilde{\omega} \right]}_{\textfrak{S}^{1,\frac{1}{2}}}
\leq 2 \norm{D^{\frac{1}{2}} X \, D^{-\frac{1}{2}}}_{\textfrak{S}^{\infty}}  \norm{\widetilde{\omega}}_{\textfrak{S}^{1,\frac{1}{2}}}$.
By the spectral theorem there exists a spectral set $\left\{ \lambda_j , \varphi_j \right\}_{j \in \mathbb{N}}$ with $\lambda_j \geq 0$ for all $j \in \mathbb{N}$ such that $\omega = \sum_{j \in \mathbb{N}} \lambda_j \ket{\varphi_j} \bra{\varphi_j}$. From 
$D^{\frac{1}{2}} \varphi_j = \lambda_j^{-1} D^{\frac{1}{2}} \omega \, \varphi_j$  we conclude
$\norm{D^{\frac{1}{2}} \varphi}_{L^2(\mathbb{R}^3)}
\leq \lambda_j^{-1} \norm{D^{\frac{1}{2}} \omega D^{\frac{1}{2}}}_{\textfrak{S}^{\infty}} \norm{D^{- \frac{1}{2}} \varphi_j}_{L^2(\mathbb{R}^3)}
\leq \lambda_j^{-1} \norm{D^{\frac{1}{2}} \omega D^{\frac{1}{2}}}_{\textfrak{S}^1} \norm{\varphi_j}_{L^2(\mathbb{R}^3)}$. This implies $\varphi \in H^{\frac{1}{2}}(\mathbb{R}^3)$, $D^{\frac{1}{2}} \omega D^{\frac{1}{2}} = \sum_{j \in \mathbb{N}} \lambda_j \ket{D^{\frac{1}{2}} \varphi_j} \bra{D^{\frac{1}{2}} \varphi_j}$, 
\begin{align}
\label{eq:Sobolev norm of omega in terms of spectral set}
\norm{\omega}_{\textfrak{S}^{1,\frac{1}{2}}} = \tr  \left( D^{\frac{1}{2}} \omega D^{\frac{1}{2}} \right) = \sum_{j \in \mathbb{N}} \lambda_j \norm{D^{\frac{1}{2}} \varphi_j}_{L^2(\mathbb{R}^3)}^2 
\quad \text{and} \quad
\rho(x) = N^{-1} \sum_{j \in \mathbb{N}} \lambda_j  \abs{\varphi_j(x)}^2 .
\end{align}
Using \eqref{eq:local well posedness finite rank system estimate for lipschitz property} we estimate for $h \in L^2(\mathbb{R}^3)$
\begin{align}
\norm{D^{\frac{1}{2}} K * \rho \, D^{- \frac{1}{2}} h}_{L^2(\mathbb{R}^3)}
&\leq N^{-1} \sum_{j \in \mathbb{N}} \lambda_j
\norm{D^{\frac{1}{2}} K * (\abs{\varphi_j}^2) \, D^{- \frac{1}{2}} h}_{L^2(\mathbb{R}^3)}
\leq C N^{-1} \abs{\gamma}  \norm{\omega}_{\textfrak{S}^{1, \frac{1}{2}}} \norm{h}_{L^2(\mathbb{R}^3)} ,
\\
\norm{D^{\frac{1}{2}} X \, D^{- \frac{1}{2}} h}_{L^2(\mathbb{R}^3)}
&\leq N^{-1} \sum_{j \in \mathbb{N}} \lambda_j
\norm{D^{\frac{1}{2}} K * (\overline{\varphi_j} \, D^{-\frac{1}{2}} h) \, \varphi_j}_{L^2(\mathbb{R}^3)}
\leq C N^{-1} \abs{\gamma}  \norm{\omega}_{\textfrak{S}^{1, \frac{1}{2}}} \norm{h}_{L^2(\mathbb{R}^3)} ,
\end{align}
which leads to
\begin{align}
\norm{i \left[ ( K * \rho - X ) , \widetilde{\omega} \right]}_{\textfrak{S}^{1,\frac{1}{2}}}
&\leq  C N^{-1} \abs{\gamma}  \norm{\omega}_{\textfrak{S}^{1,\frac{1}{2}}} \norm{\widetilde{\omega}}_{\textfrak{S}^{1,\frac{1}{2}}}. 
\end{align}
If $\omega \in \textfrak{S}^{1,\frac{1}{2}} \left( L^2(\mathbb{R}^3) \right)$ is not a positive operator we can split the compact and self adjoint operator $D^{\frac{1}{2}} \omega D^{\frac{1}{2}} = \left( D^{\frac{1}{2}} \omega D^{\frac{1}{2}} \right)_{+} - \left( D^{\frac{1}{2}} \omega D^{\frac{1}{2}} \right)_{-}$ into its positive and negative part. By its spectral decomposition one easily checks the properties 
$\left( D^{\frac{1}{2}} \omega D^{\frac{1}{2}} \right)_{+} \left( D^{\frac{1}{2}} \omega D^{\frac{1}{2}} \right)_{-} = 0$ and $\abs{D^{\frac{1}{2}} \omega D^{\frac{1}{2}}} = \left( D^{\frac{1}{2}} \omega D^{\frac{1}{2}} \right)_{+} + \left( D^{\frac{1}{2}} \omega D^{\frac{1}{2}} \right)_{-}$. The positive operators $\omega_{+} = D^{- \frac{1}{2}} \left( D^{\frac{1}{2}} \omega D^{\frac{1}{2}} \right)_{+} D^{- \frac{1}{2}}$ and
$\omega_{-} = D^{- \frac{1}{2}} \left( D^{\frac{1}{2}} \omega D^{\frac{1}{2}} \right)_{-} D^{- \frac{1}{2}}$ satisfy
$\omega_{+} - \omega_{-} = \omega$
and 
$\norm{D^s \omega D^s}_{\textfrak{S}^1} 
= \norm{D^s \omega_{+} D^s}_{\textfrak{S}^1} 
+ \norm{D^s \omega_{-} D^s}_{\textfrak{S}^1}$ .
By means of the splitting and the triangular inequality it is easily shown that
\begin{align}
\label{eq: local well-posedness Hartree auxilliary estimate for Lipschitz property}
\norm{i \left[ ( K * \rho - X ) , \widetilde{\omega} \right]}_{\textfrak{S}^{1,\frac{1}{2}}}
&\leq  C N^{-1} \abs{\gamma}  \norm{\omega}_{\textfrak{S}^{1,\frac{1}{2}}} \norm{\widetilde{\omega}}_{\textfrak{S}^{1,\frac{1}{2}}}
\end{align}
holds for all $\omega, \widetilde{\omega} \in \textfrak{S}^{1,\frac{1}{2}} \left( L^2(\mathbb{R}^3) \right)$ and that the mapping $\textfrak{S}^{1,\frac{1}{2}} \left( L^2(\mathbb{R}^3) \right) \rightarrow \textfrak{S}^{1,\frac{1}{2}} \left( L^2(\mathbb{R}^3) \right)$, $\omega \mapsto i \left[ ( K * \rho - X ), \omega \right]$ is locally Lipschitz.

\end{proof}

\begin{lemma}
\label{lemma:global well-posedness Hartree equivalence lemma}
Suppose the initial data $\omega_0$ is a finite rank operator in $\textfrak{S}_{+}^{1,\frac{1}{2}} \left( L^2(\mathbb{R}^3) \right)$ such that $0 \leq \omega_0 \leq 1$, i.e. $\omega_0 = \sum_{j=1}^M \lambda_j \psi_{j,0}(x) \overline{\psi_{j,0}(y)}$ where $\{\lambda_j \geq 0, \psi_{j,0} \}_{j=1}^M$ is a spectral set in $L^2(\mathbb{R}^3)$ with $\{ \psi_{j,0} \}_{j=1}^M \subset H^{1/2}(\mathbb{R}^3)$.
If $a = 1$ and $\gamma < 0 $ assume in addition that 
$\abs{\gamma} < \frac{N \varepsilon}{\gamma_{\rm cr} \left( \tr \left( \omega_0 \right) \right)^{2/3}}$ ,
where $\gamma_{\rm cr}$ is the universal constant in Lemma \ref{lemma: Hartree equation  orbitals global existence}. 
 Denote by $\{ \psi_{j}(t) \}_{j=1}^{M} \subset H^{1/2}(\mathbb{R}^3)$ the unique (global) solution of \eqref{eq:Hartree-Fock equation for orbitals} with initial data $\sqrt{\lambda_j} \psi_{j,0}$ given by Lemma \ref{lemma: Hartree equation  orbitals global existence}. Then $\omega_t$ with integral kernel
\begin{align}
\omega_t(x;y) 
&= \sum_{j=1}^M \psi_j(x,t) \overline{\psi_{j}(y,t)}
= \sum_{j=1}^M \lambda_j  \left( \psi_j(x,t)/ \sqrt{\lambda_j} \right) \overline{\left(\psi_{j}(y,t)/ \sqrt{\lambda_j} \right)}
\end{align}
is the unique global solution of \eqref{eq:Hartree equation relativistic integral version} in $\textfrak{S}_{+}^{1,\frac{1}{2}} \left( L^2(\mathbb{R}^3) \right)$ with initial datum $\omega_0$. 
\end{lemma}

\begin{proof}
The statement is proven in the exact same manner as \cite[Proposition 2.4]{C1976}.

\end{proof}

\begin{proof}[Proof of Proposition~\ref{lemma: Hartree equation global existence of solution}]
%The Lemma is proven in a similar way as \cite[Theorem 2.5]{C1976}. 
Note that $\omega_0 \in \textfrak{S}^{1,\frac{1}{2}}_{+} \left( L^2(\mathbb{R}^3) \right)$ can be written as $\omega_0 = \sum_{j=1}^{\infty} \lambda_j \ket{\psi_{j,0}} \bra{\psi_{j,0}}$
where $\{ \lambda_{j} \geq 0, \psi_{j,0} \}_{j \in \mathbb{N}}$ is a spectral set and $$\norm{\omega_0}_{\textfrak{S}^{1,\frac{1}{2}}} = \sum_{j=1}^{\infty} \lambda_j  \norm{D^{1/2} \psi_{j,0}}_{L^2(\mathbb{R}^3)}^2 
= \sum_{j=1}^{\infty} \lambda_j \norm{\psi_{j,0}}_{H^{1/2}(\mathbb{R}^3)}^2  < + \infty.$$ Then
$\left\{ \omega^{\leq M}_0 = \sum_{j=1}^{M} \lambda_j \ket{\psi_{j,0}} \bra{\psi_{j,0}} \right\}_{M=1}^{\infty}$ is a sequence of finite rank operators approximating $\omega_0$ in $\textfrak{S}^{1, \frac{1}{2}} \left( L^2(\mathbb{R}^3) \right)$. 
Let $\{ \psi_{j}(t) \}_{j=1}^M \subset H^{1/2}(\mathbb{R}^3)$ be the unique (global) solution of \eqref{eq:Hartree-Fock equation for orbitals} with initial data $\{ \sqrt{\lambda_j} \psi_{j,0} \}_{j=1}^M \subset H^{1/2}(\mathbb{R}^3)$.
According to Lemma \ref{lemma:global well-posedness Hartree equivalence lemma}
\begin{align}
\omega^{\leq M}_t
&= \sum_{j=1}^M \ket{\psi_j(t)} \bra{\psi_j(t)}
= \sum_{j=1}^M \lambda_j \ket{\big(\psi_j(t) / \sqrt{\lambda_j} \big)} \bra{\big( \psi_j(t) / \sqrt{\lambda_j} \big)}
\end{align}
is the unique global solution of \eqref{eq:Hartree equation relativistic integral version} with data $\omega^{\leq M}_0$ at $t = 0$. In the following we show that $\lim_{M \rightarrow \infty} \omega^{\leq M}_t$ converges in $\textfrak{S}^{1,\frac{1}{2}} \left( L^2(\mathbb{R}^3) \right)$ and that the limiting operator is a solution of \eqref{eq:Hartree equation relativistic integral version}.
Let $t \in (0, \infty)$ and $L,M \in \mathbb{N}$ such that $L \geq M$. To this end note that 
\begin{align}
\norm{\omega^{\leq L}_t - \omega^{\leq M}_t}_{\textfrak{S}^{1,\frac{1}{2}}}
&\leq 
\varepsilon^{-1} 
\tr \left( \sqrt{1 - \varepsilon^2 \Delta} \left(\omega^{\leq L}_t - \omega^{\leq M}_t \right) \right) .
\end{align}
Using once more the conservation of the energy \eqref{eq:definition energy finite rank Hartree system} we write the right-hand side as
\begin{align}
\tr \left( 
\sqrt{1 - \varepsilon^2 \Delta} \left( \omega^{\leq L}_t - \omega^{\leq M}_t \right)
\right)
&= \tr \left( 
\sqrt{1 - \varepsilon^2 \Delta} \left( \omega^{\leq L}_0 - \omega^{\leq M}_0 \right)
\right)
\nonumber \\
&\quad + \frac{1}{2}
\left[ 
\tr \left( K * \rho_0^{\leq L} \, \omega^{\leq L}_0 \right)
-
\tr \left( K * \rho_0^{\leq M} \, \omega^{\leq M}_0 \right) 
\right]
\nonumber \\
&\quad - \frac{1}{2}
\left[ 
\tr \left( X_0^{\leq L} \, \omega^{\leq L}_0 \right)
-
\tr \left( X_0^{\leq M} \, \omega^{\leq M}_0 \right) 
\right]
\nonumber \\
&\quad - \frac{1}{2}
\left[ 
\tr \left( K * \rho_t^{\leq L} \, \omega^{\leq L}_t \right)
-
\tr \left( K * \rho_t^{\leq M} \, \omega^{\leq M}_t \right) 
\right] 
\nonumber \\
&\quad + \frac{1}{2}
\left[ 
\tr \left( X_t^{\leq L} \, \omega^{\leq L}_t \right)
-
\tr \left( X_t^{\leq M} \, \omega^{\leq M}_t \right) 
\right]  .
\end{align}
In the following, we drop the $t$-dependence to simplify the notation. By   \eqref{eq:Hatree local well-posedness auxiliary estimates 3} and the fact that $\omega^{\leq M} \leq \omega^{\leq L}$ we get
\begin{align}
& \Big| \tr \left( K * \rho^{\leq L} \, \omega^{\leq L} \right)
-
\tr \left( K * \rho^{\leq M} \, \omega^{\leq M} \right) \Big|
\nonumber \\
&\quad \leq N  \int_{\mathbb{R}^{3}} dx \, K * \rho^{\leq L}(x) \left( \rho^{\leq L}(x) - \rho^{\leq M}(x) \right) 
+ N   \int_{\mathbb{R}^{3}}  K *\rho^{\leq M}(x) \left( \rho^{\leq L}(x) - \rho^{\leq M}(x) \right)dx 
\nonumber \\
&\quad \leq C N^{-1} \varepsilon^{-a} \left( \tr \left( \sqrt{ - \varepsilon^2 \Delta} \left( \omega^{\leq L} - \omega^{\leq M} \right) \right) \right)^{\frac{a}{2}}
\left( \tr \left( \omega^{\leq L} - \omega^{\leq M} \right) \right)^{\frac{2 - a}{2}}
\nonumber \\
&\qquad \qquad \times
\left[ 
\left( \tr \left( \sqrt{- \varepsilon^2 \Delta} \omega^{\leq L} \right) \right)^{\frac{a}{2}} \left( \tr \left( \omega^{\leq L} \right) \right)^{\frac{2 - a}{2}}
+
\left( \tr \left( \sqrt{- \varepsilon^2 \Delta} \omega^{\leq M} \right) \right)^{\frac{a}{2}} \left( \tr \left( \omega^{\leq M} \right) \right)^{\frac{2 - a}{2}}
\right]
\nonumber \\
&\quad \leq   \left( \tr \left( \sqrt{ 1 - \varepsilon^2 \Delta} \left( \omega^{\leq L} - \omega^{\leq M} \right) \right) \right)^{a}
\nonumber \\
&\qquad +
C  N^{-2} \varepsilon^{-2a}  \left( \tr \left( \omega^{\leq L} - \omega^{\leq M} \right) \right)^{2 - a}
\left( \tr \left( \sqrt{- \varepsilon^2 \Delta} \omega^{\leq L} \right) \right)^{a} \left( \tr \left( \omega^{\leq L} \right) \right)^{2 - a} .
\end{align}
Since
\begin{align}
&\abs{\tr \left( X^{\leq L} \omega^{\leq L} \right) - \tr \left( X^{\leq M} \omega^{\leq M} \right) }
\nonumber \\
&\quad \leq 
\abs{\tr \left( \left( X^{\leq L} - X^{\leq M} \right) \omega^{\leq L} \right)} 
+ \abs{\tr \left( X^{\leq M} \left( \omega^{\leq L} -  \omega^{\leq M} \right) \right) }
\nonumber \\
&\quad \leq N  \int_{\mathbb{R}^{3}}  K * \rho^{\leq L}(x) \left( \rho^{\leq L}(x) - \rho^{\leq M}(x) \right) dx
+ N   \int_{\mathbb{R}^{3}}  K *\rho^{\leq M}(x) \left( \rho^{\leq L}(x) - \rho^{\leq M}(x) \right)dx 
\end{align}
holds because of \eqref{eq:Hatree local well-posedness auxiliary estimates 5}
we can estimate the exchange term by the same means and obtain
\begin{align}
&\tr \left( 
\sqrt{1 - \varepsilon^2 \Delta} \left( \omega^{\leq L}_t - \omega^{\leq M}_t \right)
\right)
\nonumber \\
&\quad \leq C \tr \left( 
\sqrt{1 - \varepsilon^2 \Delta} \left( \omega^{\leq L}_0 - \omega^{\leq M}_0 \right)
\right)
\\
&\qquad + 
C N^{-2}  \varepsilon^{-2a} {\sup_{\tau \in \{0,t \}}}  \left( \tr \left( \omega^{\leq L}_\tau - \omega^{\leq M}_\tau \right) \right)^{2 - a}
\left( \tr \left( \sqrt{- \varepsilon^2 \Delta} \omega^{\leq L}_\tau \right) \right)^{a} \left( \tr \left( \omega^{\leq L}_\tau \right) \right)^{2 - a}.
\end{align}
The conservation $\tr \left( \omega^{\leq L}_t \right) = \tr \left( \omega^{\leq L}_0 \right)$, and $\norm{\omega_0}_{\textfrak{S}^{1,\frac{1}{2}}} < \infty$ imply
\begin{align}
\tr \left( 
\sqrt{1 - \varepsilon^2 \Delta} \left( \omega^{\leq L}_0 - \omega^{\leq M}_0 \right)
\right) 
&\leq \varepsilon^{-1} \sum_{j=M}^L \lambda_j \norm{\psi_{j,0}}_{H^{1/2}(\mathbb{R}^3)}^2 \rightarrow 0 ,
\nonumber \\
\tr \left( \omega^{\leq L}_t - \omega^{\leq M}_t \right)
&= \tr \left( \omega^{\leq L}_0 - \omega^{\leq M}_0 \right)
\rightarrow 0 
\end{align}
as $M, L \rightarrow \infty$. Since (see sketch of proof of Lemma \ref{lemma: Hartree equation  orbitals global existence})
\begin{align}
\tr \left( \sqrt{- \varepsilon^2 \Delta} \omega^{\leq L}_t \right)
&\leq \varepsilon \norm{\{ \psi_j(t) \}_{j=1}^L}_{H^{1/2,L}}^2 
\nonumber \\
&\leq C(N , \varepsilon)  \norm{\{ \psi_{j}(0) \}_{j=1}^L}_{H^{1/2,L}}^2 
\nonumber \\
&= C(N , \varepsilon) \sum_{j=1}^L \lambda_j \norm{ \psi_{j,0} }_{H^{1/2} (\mathbb{R}^3)}^2
\nonumber \\
&\leq C(N , \varepsilon) \norm{\omega_0}_{\textfrak{S}^{1, \frac{1}{2}}} < + \infty
\end{align}
we obtain
\begin{align}
\norm{\omega^{\leq L}_t - \omega^{\leq M}_t}_{\textfrak{S}^{1,\frac{1}{2}}}
&\leq 
\varepsilon^{-1} 
\tr \left( \sqrt{1 - \varepsilon^2 \Delta} \, \left( \omega^{\leq L}_t - \omega^{\leq M}_t \right) \right)
\rightarrow 0 
\quad \text{as} \; M, L \rightarrow \infty .
\end{align}
We consequently have that $\{ \omega^{\leq M} \}_{M \in \mathbb{N}}$ is a Cauchy sequence in $\textfrak{S}^{1, \frac{1}{2}}$ uniformly in $t \in (0, \infty)$. Hence, it converges to an operator $\omega_t$. The operator is continuous in $t$ because of the uniform limit theorem and the fact that
 $\omega^{\leq M}_t$ is continuous in $t$ (see Lemma \ref{lemma: Hartree equation  orbitals global existence}).
From the fact that $\omega^{\leq M}_t$ is a solution of \eqref{eq:Hartree equation relativistic integral version} and
\begin{align}
& \lim_{M \rightarrow \infty} \norm{\omega^{\leq M}_t - \omega_t}_{\textfrak{S}^{1, \frac{1}{2}}} = 0 ,
\nonumber \\
&\lim_{M \rightarrow \infty} \norm{e^{- i \sqrt{1 - \varepsilon^2 \Delta} t} \left( \omega^{\leq M}_0 - \omega_0 \right) e^{i \sqrt{1 - \varepsilon^2 \Delta} t} }_{\textfrak{S}^{1, \frac{1}{2}}} = \norm{\omega^{\leq M}_0 - \omega_0}_{\textfrak{S}^{1, \frac{1}{2}}} = 0 ,
\nonumber \\
& \lim_{M \rightarrow \infty}
\norm{\left[ K * \rho_s^{\leq M} , \omega^{\leq M}_s \right] - \left[ K * \rho_s , \omega_s \right]}_{\textfrak{S}^{1, \frac{1}{2}}}
\leq  C \lim_{M \rightarrow \infty}
\left( \norm{\omega^{\leq M}_s}_{\textfrak{S}^{1, \frac{1}{2}}} + \norm{\omega_s}_{\textfrak{S}^{1,\frac{1}{2}}} \right) \norm{\omega^{\leq M}_s - \omega_s}_{\textfrak{S}^{1, \frac{1}{2}}}
\end{align}
for all $t, s \in (0, \infty)$
it directly follows that $\omega_t$ satisfies \eqref{eq:Hartree equation relativistic integral version}.
The uniqueness of the global solution follows from the uniqueness of the local solution (recall Lemma \ref{lemma: Hartree equation local existence of solution}). Since $\textfrak{S}^{1,\frac{1}{2}}_{+}\left( L^2 (\mathbb{R}^3) \right)$ is a closed subspace of $\textfrak{S}^{1,\frac{1}{2}} \left( L^2 (\mathbb{R}^3) \right)$ and $\{\omega^{\leq M}_t \}_{M \in \mathbb{N}}$ is a positive sequence for all $t \in (0, \infty)$ by construction we have that $\omega_t$ is positive for all $t \in (0, \infty)$.
\end{proof}

\section{Rigorous Duhamel expansion}\label{appendix:duhamel}

In this section, we give the details on how one obtains the Duhamel expansion from Section \ref{subsection:Derivation of the relativistic Vlasov equation}.

\begin{proof}[Derivation of \eqref{eq:derivation Vlasov from Hartree Duhamel 1}--\eqref{eq:derivation Vlasov from Hartree Duhamel 3} and \eqref{eq:Duhamel expansion for general Schatten spaces}]

Let $\Lambda > 1$ and $K_{\Lambda}: \mathbb{R}^d \rightarrow \mathbb{R}$ be a potential (whose explicit form will be chosen later) such that $t \mapsto \norm{K_\Lambda * \rho_t}_{\textfrak{S}^\infty}$ is a strongly continuous map of $\mathbb{R}$ into the bounded self-adjoint operators.
In analogy to \cite[Chapter 3]{DRS2017} we define the two parameter group $U_{\Lambda}(t;s)$ satisfying
\begin{align}
i \varepsilon \partial_t U_{\Lambda}(t;s)
&= \left( \sqrt{1 - \varepsilon^2 \Delta} + K_\Lambda* \rho_t \right) U_{\Lambda}(t;s)
\quad \text{and}
\; \; U_{\Lambda}(s;s) = 1
\end{align}
by means of the interaction picture. Using Duhamel's formula we obtain
\begin{align}
U^*_{\Lambda}(t;0) \left( \omega_{N,t} - \widetilde{\omega}_{N,t} \right) U_{\Lambda}(t;0)
&= \omega_{N,0} - \widetilde{\omega}_{N,0}
\nonumber \\
&\quad 
- \frac{i}{\varepsilon} \int_0^t 
U_{\Lambda}^*(s;0) 
\left( 
\left[ \sqrt{1 - \varepsilon^2 \Delta} , \widetilde{\omega}_{N,s} \right] - A_{N,s} \right) U_{\Lambda}(s;0)\,ds
\nonumber \\
&\quad 
- \frac{i}{\varepsilon} \int_0^t
U_{\Lambda}^*(s;0) 
\left( 
\left[ K * \rho_s , \widetilde{\omega}_{N,s} \right] - B_{N,s}
\right) U_{\Lambda}(s;0)\,ds
\nonumber \\
&\quad 
- \frac{i}{\varepsilon} \int_0^t 
U_{\Lambda}^*(s;0) 
\left[ \left( K - K_{\Lambda} \right) * \rho_s , \left( \omega_{N,s} - \widetilde{\omega}_{N,s} \right)  \right] U_{\Lambda}(s;0)\,ds .
\end{align}
Hence,
\begin{align}
\label{eq:Duhamel expansion for general Schatten spaces appendix} 
\norm{\omega_{N,t} - \widetilde{\omega}_{N,t}}_{\textfrak{S}^p}
&\leq   \norm{\omega_{N} - \widetilde{\omega}_{N}}_{\textfrak{S}^p}
\nonumber \\
&\quad + 
\frac{1}{\varepsilon} \int_0^t 
\norm{\left[ \sqrt{1 - \varepsilon^2 \Delta} , \widetilde{\omega}_{N,s} \right] - A_{N,s} }_{\textfrak{S}^p} \,ds
\nonumber \\
&\quad 
+ \frac{1}{\varepsilon} \int_0^t
\norm{\left[ K * \rho_s , \widetilde{\omega}_{N,s} \right] - B_{N,s}}_{\textfrak{S}^p}
\,ds
\nonumber \\
&\quad
+ \frac{1}{\varepsilon}
\int_0^t 
\norm{\left[ \left( K - K_{\Lambda} \right) * \rho_s , \omega_{N,s} - \widetilde{\omega}_{N,s}   \right]}_{\textfrak{S}^p}
\,ds .
\end{align}
In the following we show for suitable chosen $K_{\Lambda}$ that the last term on the right-hand side converges to zero if we take  the limit $\Lambda \rightarrow \infty$. This shows the claim.\\ We study separately the following two  cases:

\paragraph{Case $d = 3$ and $a \in (0,1]$:}
Note that $\omega_t \in \textfrak{S}_{+}^{1,\frac{1}{2}} \left( L^2(\mathbb{R}^3) \right)$ holds locally in time, respectively globally if $\abs{\gamma}$ is small enough, because of the initial conditions of Theorem \ref{theorem:derivation relativistic Vlasov from Hartree}
and Proposition \ref{lemma: Hartree equation global existence of solution}. Since
\begin{align}
\label{eq:L-infty estimate for the potential}
\norm{K * \rho_s}_{L^{\infty}(\mathbb{R}^3)}
&\leq   C \abs{\gamma} N^{-1} \norm{\omega_s}_{\textfrak{S}^{1,\frac{1}{2}}}
< + \infty 
\end{align}
it is possible to choose $K_{\Lambda} = K$, implying that the last summand on  the right hand side of \eqref{eq:Duhamel expansion for general Schatten spaces appendix}  equals zero.
In order to show \eqref{eq:L-infty estimate for the potential} consider the spectral set $\{\lambda_j, \varphi_{j} \}_{j \in \mathbb{N}}$ with $\lambda_j \geq 0$ of an operator $\omega \in \textfrak{S}_{+}^{1,\frac{1}{2}} \left( L^2(\mathbb{R}^3) \right)$ and recall \eqref{eq:Sobolev norm of omega in terms of spectral set}. Together with \eqref{eq:L-infty estimate on convoluted potential} we obtain
\begin{align}
\norm{K * \rho}_{L^{\infty}(\mathbb{R}^3)}
&\leq  N^{-1} \sum_{j \in \mathbb{N}} \lambda_j \norm{K * \abs{\varphi_j}^2}_{L^{\infty}(\mathbb{R}^3)}
\leq C \abs{\gamma} N^{-1} \sum_{j \in \mathbb{N}} \lambda_j \norm{\varphi_j}_{H^{1/2}(\mathbb{R}^3)}^2 
= C \abs{\gamma} N^{-1} \norm{\omega}_{\textfrak{S}^{1,\frac{1}{2}}}.
\end{align}

\paragraph{Case $d \in \{ 2, 3 \}$ and $a \in \left( \max \left\{\frac{d}{2} - 2, -1 \right\} , 0 \right]$:}
Let $K_{\Lambda}(x) = K(x) \id_{\abs{x} \leq \Lambda}$. Then
\begin{align}
\norm{\left(K - K_{\Lambda} \right) * \rho_s \; \left(1 + \abs{x} \right)^{-1}}_{L^{\infty}(\mathbb{R}^d)}
&\leq \Lambda^{\abs{a} - 1}
\sup_{x \in \mathbb{R}^d}
\left\{
\left(1 + \abs{x} \right)^{-1} \int_{\abs{x-y} > \Lambda}  \abs{x-y} \rho_s(y)\,dy
\right\}
\nonumber \\
&\leq \Lambda^{\abs{a} - 1}
\left( \norm{\rho_s}_{L^1(\mathbb{R}^d)}
+ N^{-1} \tr \left( \abs{x}  \omega_{N,s} \right) \right)
\nonumber \\
&\leq \Lambda^{\abs{a} - 1}
\left( 1 + N^{-1/2} \sqrt{\tr \left( x^2 \, \omega_{N,s} \right) } \right).
\end{align}
To obtain the ultimate inequality we have used the estimate
$
\abs{\tr \left( \abs{x} \omega_{N,s} \right)}
\leq \norm{\abs{x} \sqrt{\omega_{N,s}}}_{\textfrak{S}^2} \norm{\sqrt{\omega_{N,s}}}_{\textfrak{S}^2}
$.
Together with the embedding $\norm{\omega}_{\textfrak{S}^p} = 2^{\frac{1}{p}}\norm{\omega}_{\textfrak{S}^1}$ for $1 \leq p < \infty$ and
$
\norm{\left( 1 + \abs{x} \right) \left( 1 + x^2 \right)^{- 3/2} \left( 1 - \varepsilon^2 \Delta \right)^{-1} }_{\textfrak{S}^2} \leq C \sqrt{N}
$
this leads to
\begin{align}
\label{eq: rigorous duhamel expansion intermediate estimate}
&\norm{\left[ \left( K - K_{\Lambda} \right) * \rho_s , \omega_{N,s} - \widetilde{\omega}_{N,s}   \right]}_{\textfrak{S}^p}
\nonumber \\
&\quad \leq C 
\norm{\left(K - K_{\Lambda} \right) * \rho_s \; \left(1 + \abs{x} \right)^{-1}}_{L^{\infty}(\mathbb{R}^d)}
\left( 
\norm{\left(1 + \abs{x} \right) \, \omega_{N,s}}_{\textfrak{S}^1}
+ \norm{\left(1 + \abs{x} \right) \, \widetilde{\omega}_{N,s}}_{\textfrak{S}^1}
\right)
\nonumber \\
&\quad \leq C \,\Lambda^{\abs{a} - 1}
\left(  N^{1/2} + \sqrt{\tr \left( x^2 \, \omega_{N,s} \right) } \right) 
\left( 
N 
+ \sqrt{\tr \left( x^2 \, \omega_{N,s} \right) }
+ \norm{ \left( 1 - \varepsilon^2 \Delta \right) \left( 1 + x^2 \right)^{3/2} \widetilde{\omega}_{N,s}}_{\textfrak{S}^2}
\right) .
\end{align}
Using
$
\left[ \sqrt{1 - \varepsilon^2 \Delta} , x \right]
= - \frac{ \varepsilon^2 \nabla}{\sqrt{1 - \varepsilon^2 \Delta}}
$, $\norm{\frac{i \varepsilon \nabla}{\sqrt{1 - \varepsilon^2 \Delta}}}_{\textfrak{S}^{\infty}} \leq 1$, $\tr \left( \omega_{N,t} \right) = N$ and the cyclicity of the trace we get
\begin{align}
i \frac{d}{dt} \tr \left( x^2 \omega_{N,t} \right)
&=  \tr \left( \left[  x^2 , \sqrt{1 - \varepsilon^2 \Delta}  \right] \omega_{N,t}  \right)
\nonumber \\
&= -  \tr \left( \left\{ \frac{i \varepsilon^2 \nabla}{\sqrt{1 - \varepsilon^2 \Delta}} , x \right\} \omega_{N,t}  \right)
\nonumber \\
&\leq 2 \varepsilon \norm{x \sqrt{\omega_{N,t}}}_{\textfrak{S}^2}
\norm{\sqrt{\omega_{N,t}} \frac{i \varepsilon \nabla}{\sqrt{1 - \varepsilon^2 \Delta}}}_{\textfrak{S}^2}
\nonumber \\
&\leq 2 N^{1/2} \sqrt{\tr \left( x^2 \omega_{N,t} \right)}.
\end{align}
For the initial data of Theorem \ref{theorem:derivation relativistic Vlasov from Hartree} it consequently holds that $\abs{\tr \left( x^2 \omega_{N,t} \right)} < + \infty$ for all $t \in \mathbb{R}$. 
By similar estimates as in \cite[Chapter 3]{BPSS2016} one derives 
\begin{align}
\norm{ \left( 1 - \varepsilon^2 \Delta \right) \left( 1 + x^2 \right)^{3/2} \widetilde{\omega}_{N,s}}_{\textfrak{S}^2}
&\leq C N^{1/2} \norm{\widetilde{W}_{N,s}}_{H_5^5(\mathbb{R}^{2d})} .
\end{align}
Since the right-hand side is finite for $\widetilde{\omega}_{N,s}$ from 
Theorem \ref{theorem:derivation relativistic Vlasov from Hartree} we obtain
\begin{align}
\lim_{\Lambda \rightarrow \infty} 
\norm{\left[ \left( K - K_{\Lambda} \right) * \rho_s , \omega_{N,s} - \widetilde{\omega}_{N,s}   \right]}_{\textfrak{S}^p}
= 0
\end{align}
by means of inequality \eqref{eq: rigorous duhamel expansion intermediate estimate}.

\end{proof}

\noindent{\it Acknowledgments.} The authors thank Daniele Dimonte for many constructive discussions on the results contained in this paper and related topics, and Enno Lenzmann for pointing out reference \cite{HLLS2010}.\\ 
N.L. and C.S. gratefully acknowledge support from the Swiss National Science Foundation through the SNSF Eccellenza project PCEFP2 181153 and the NCCR SwissMAP. N.L., in addition,  gratefully acknowledges funding from the European Union's Horizon 2020 research and innovation programme through the Marie Sk{\l}odowska-Curie Action EFFECT (grant agreement No. 101024712).

%\bibliographystyle{plain}
%\bibliography{references}

{}

\end{document}